\documentclass{llncs}
\usepackage{amsmath,amssymb,cite}
\usepackage{fullpage}
\usepackage[hcentering]{geometry}
\usepackage{caption} \usepackage{subcaption}
\usepackage{graphicx}
\usepackage{float}

\usepackage{paralist}
\usepackage{url}

\newcommand{\hide}[1]{}

\DeclareMathOperator{\Real}{\mathbb{R}}

\usepackage[colorlinks=true, linkcolor=red, urlcolor=blue, citecolor=gray]{hyperref}
\usepackage{color}
\usepackage[usenames,dvipsname s,svgnames,table]{xcolor}
  \usepackage{nth}
  \usepackage{intcalc}

\begin{document}
\title{A Basic Compositional Model for Spiking Neural
  Networks\thanks{This work was supported by NSF awards CCF-1810758, CCF-2139936,
    CCF-2003830, CCF-2046235, and CCF-1763618.  Musco was also
    partially supported by an Adobe Research grant.}}
\titlerunning{Compositional SNN Model}

\author{Nancy Lynch\inst{1} \and Cameron Musco\inst{2}}
\institute{
  Department of Electrical Engineering and Computer Science,
  Massachusetts Institute of Technology \and Department of Computer
  Science, University of
  Massachusetts, Amherst}
\maketitle

\begin{abstract}
We present a formal, mathematical foundation for modeling and
reasoning about the behavior of \emph{synchronous, stochastic Spiking Neural
Networks (SNNs)}, which have been widely used in studies of neural
computation.
Our approach follows paradigms established in the field of concurrency theory.

Our SNN model is  based on directed graphs of neurons, classified as
input, output, and internal neurons.
We focus here on basic SNNs, in which a neuron's only state is a Boolean
value indicating whether or not the neuron is currently firing.
We also define the \emph{external behavior} of an SNN, in terms of
probability distributions on its external firing patterns.
We define two operators on SNNs:  a \emph{composition operator}, which
supports modeling of SNNs as combinations of smaller SNNs, and a 
\emph{hiding operator}, which reclassifies some output behavior of an
SNN as internal.
We prove results showing how the external behavior of a network
built using these operators is related to the external behavior of its
component networks.
Finally, we define the notion of a \emph{problem} to be solved
by an SNN, and show how the composition and
hiding operators affect the problems that are solved by the networks.

We illustrate our definitions with three examples:  
a Boolean circuit constructed from gates, 
an \emph{Attention} network constructed from a \emph{Winner-Take-All}
network and a \emph{Filter} network, and 
a toy example involving combining two networks in a cyclic fashion.
 \end{abstract}

\section{Introduction}


%
Understanding computation in biological neural networks like the human
brain is a central challenge of modern neuroscience and artificial
intelligence.
One approach to this challenge uses algorithmic methods from
theoretical computer science.
That means defining formal computational models for brain networks,
identifying abstract problems that can be solved by such networks, and
defining and analyzing algorithms that solve these problems.
Work along these general lines includes that of Valiant, Navlakha,
Papadimitriou, and their collaborators (see, for example, \cite{Valiant00,DasguptaSN17,PapadimitriouVempala-itcs19}).

%
For the past few years, we and our collaborators have been working on
an algorithmic theory of brain networks, based on \emph{synchronous,
  stochastic Spiking Neural Network (SNN) models}.
SNNs are a model for neural computation that includes many important
biologically-plausible features, yet is still simple enough to study
theoretically.
An SNN is a directed graph of neurons, in which each neuron fires
in discrete spikes, in response to a sufficiently high membrane potential.
The potential is induced by spikes from neighboring neurons, which can
be either excitatory or inhibitory, increasing or decreasing the
incoming potential.
In our SNNs, the neurons operate in synchronous rounds, and make
firing decisions stochastically.
%
%
Inspired by tasks that are solved in actual brains, we have been defining
and studying abstract problems to be solved by our SNNs.
So far,  we have developed models and networks for the
\emph{Winner-Take-All} problem from computational
neuroscience~\cite{LynchMP16,muscoThesis,LynchMP-arxiv19,SuCL-jour19},
problems of neural coding and similarity
detection~\cite{LynchMP16b,HitronMPL-itcs20},
problems of spatial representation of temporal information~\cite{HitronP-esa19,WL19},
and problems involving
learning~\cite{ChouWang-colt20,ChouWY-arxiv20,Wang20,LynchMallmannTrenn21}.
We are continuing to study many other problems and networks, including both
static networks and networks that learn.

%
In our work so far, we have defined formal models in each paper, as needed.
Here we define a more general computational model for SNNs that
we hope will provide a useful foundation for formal modeling of many
networks and formal reasoning about their behavior.  
Note that this model is not the most general one that will be needed,
but we believe that it will prove to be a useful first step.
In particular, in the basic version of the SNN model defined here, a
neuron's only state is a Boolean value indicating whether or not the neuron is
currently firing.
This is sufficient to model some algorithms, such as the 
simple two-inhibitor Winner-Take-All network in~\cite{LynchMP-arxiv19}.
Other algorithmic work uses variants of the basic model with more 
elaborate state such as limited local history, or flags that enable
certain behavior such as learning~\cite{SuCL-jour19,
  LynchMallmannTrenn21}; we expect that the results of
this paper should be extendable to these variants as well, but this remains
to be worked out.
We also define an \emph{external behavior notion} for SNNs, in terms
of probability distributions on its external firing patterns.  This can
be used for stating requirements to be satisfied by the networks.

%
We then define a \emph{composition operator} for SNNs, which supports
modeling of SNNs as combinations of smaller SNNs.
We prove that our external behavior notion is \emph{compositional}, in
the sense that the external behavior of a composed network depends
only on the external behaviors of the component networks and not their
internal operation.
We also define a \emph{hiding operator} that reclassifies some output
behavior of an SNN as internal, and show that the behavior of a network
obtained by hiding depends only on that of the original network.
A common use of hiding is after composition, when some of the
interactions between the composed networks might be suppressed in the
external behavior.

%
Finally, we give a formal definition of a \emph{problem} to be solved
by an SNN, and give basic results showing how the composition and
hiding operators affect the problems that are solved by the networks.
We illustrate our definitions with three examples:  
a Boolean circuit constructed from neurons that act as logical gates, 
an \emph{Attention} network constructed from a \emph{Winner-Take-All}
network and a \emph{Filter} network, and 
a toy example involving combining two networks in a cyclic fashion.

\vspace{-.25cm}
\paragraph{Related work:}
The general approach of this paper---defining formal models and
operators and proving that the operators respect network behavior---is
based on the paradigms of the research area of \emph{concurrency
theory}~\cite{DBLP:conf/concur/2021}.
Our particular definitions are inspired by prior work on 
Input/Output Automata models~\cite{LynchT87,LynchT89,LV93I,LV95II-jour,SL95,LSV02,LSVprob,KLSV10}, including timed,
hybrid and probabilistic variants.

Our focus on SNNs was partly inspired by research of Maass, et al.~\cite{maass1997networks,maass1999neural,maass2000computational}
on the computational power of SNNs.  Maass explored how features like
randomness~\cite{maass2014noise}, temporal coding~\cite{maass1999complexity}, and dynamic edge
weights~\cite{legenstein2005can} affect the computational power and efficiency of
neural network models.  Maass's work differs from ours in that he mostly
considers asynchronous models that allow fine-grained control of spike
timing---models with significantly different computing power
from ours.

An early synchronous neural network model is the \emph{perceptron
model}, based on a neuron model invented by McCulloch and
Pitts~\cite{mcculloch43a}.
The neurons are modeled as deterministic linear threshold elements,
without any stochastic behavior as in our neurons.
These elements are assembled into feedforward, layered networks,
whereas our networks are arbitrary directed graphs.
Another difference with respect to our basic model is that, in
perceptron networks, real values can be passed along edges between
layers, whereas we use a binary activation function.  
Perceptron networks are generally used to implement supervised
algorithms for learning to recognize patterns.

Work by Valiant, Navlakha, Papadimitriou, and
collaborators~\cite{Valiant00,DasguptaSN17,PapadimitriouVempala-itcs19},
is based on a variety of synchronous neural network models.  These
models are not presented as general compositional models in the style
of concurrency theory.
However, they appear to be compatible with (extended versions of) our
model.
Some differences between these models and our basic model are:
Valiant~\cite{Valiant00} includes elaborate state changes, rather than
just simple binary firing decisions;
Papadimitriou\cite{PapadimitriouVempala-itcs19} and
Navlakha~\cite{DasguptaSN17}
assume built-in Winner-Take-All mechanisms; and Valiant and Papadimitriou
focus on learning.

In recent work, Berggren and his group are developing a
hardware implementation of a Spiking Neural Network model using
nanowires~\cite{Toomey,ToomeySCCLB}.
They have developed a simulator for their
implementation, based on the basic SNN model presented in this paper
\cite{LomboLCCLB}.

\vspace{-.25cm}
\paragraph{Paper organization:}
Section~\ref{sec: model} contains our definitions for Spiking Neural
Networks and their external behavior.
Section~\ref{sec: composition} contains our definitions for the
composition operator for SNNs. 
Section~\ref{sec: acyclic} focuses on the special case of acyclic
composition, in which connections between SNNs go in only one
direction; we prove a compositionality theory for this case.
Section~\ref{sec: general} extends these ideas to the more general case of
composition that allows connections in both directions.
Section~\ref{sec: hiding} introduces the hiding operator for SNNs.
Section~\ref{sec: problem} introduces our notion of a problem to be
solved by an SNN.
We conclude in Section~\ref{sec: conclusions}.

\vspace{-.25cm}
\paragraph{Acknowledgments:}
We thank our co-authors on our papers based on SNNs, especially
Merav Parter, Lili Su, Brabeeba Wang, Yael Hitron, C. J. Chang, and Frederik
Mallmann-Trenn, for providing many concrete examples that inspired our
general model.  We also thank Victor Luchangco and Jesus Lares for
reading and contributing comments on earlier drafts of this paper.

\section{The Spiking Neural Network Model}
\label{sec: model}

Here we present our model definitions.
We first specify the structure of our networks---the neurons and
connections between them.
Then we describe how the networks execute; this involves defining
individual (non-probabilistic) executions and then defining
probabilistic behavior.
Next we define the external behavior of a network.
We illustrate with two fundamental examples:  a Boolean circuit and a
Winner-Take-All network.

\subsection{Network structure}

Assume a universal set $U$ of neuron names.
A \emph{firing pattern} for a set $V \subseteq U$ of neuron names is a
mapping from $V$ to $\{0,1\}$.  Here, $1$ represents ``firing'' and $0$
represents ``not firing''.

A \emph{Spiking Neural Network}, which we generally refer to as just a
\emph{network}, ${\mathcal N}$, consists of:
\begin{itemize}
\item
$N$, a subset of $U$, partitioned into input neurons $N_{in}$, output
  neurons $N_{out}$, and internal
  neurons $N_{int}$.
  We sometimes write $N_{ext}$ as shorthand for $N_{in} \cup N_{out}$,
  and $N_{lc}$ as shorthand for $N_{out} \cup N_{int}$.  (Here, $lc$
  stands for ``locally controlled'', which means ``not input'').
%
Each neuron $u \in N_{lc}$ has an associated \emph{bias}, $bias(u) \in
\Real$; this can be any real number, positive, negative, or 0.

\item
$E$, a set of ordered pairs of neurons, i.e., directed edges between
  neurons, representing synapses.
We permit self-loops.
Each edge $e = (u,v)$ has a \emph{weight}, $weight(e) = weight(u,v)$,
which is a nonzero (positive or negative) real number.  

\item
$F_0$, an initial firing pattern for the set $N_{lc}$ of non-input
  neurons; that is, $F_0: N_{lc} \rightarrow \{0,1\}$.
\end{itemize}
We assume that input neurons have no incoming edges, not even
self-loops.  Output neurons may have incoming or outgoing edges, or
both.

\vspace{-.25cm}
\paragraph{Example:}
Consider the Winner-Take-All network in Figure~\ref{fig:wta}.
The set $N$ of neuron names consists of $N_{in} = \{
x_1,\ldots,x_n \}$, $N_{out} = \{ y_1, \ldots, y_n \}$, and $N_{int} =
\{ a_1, a_2 \}$.
We have $bias(a_1) = .5 \gamma$, $bias(a_2) = 1.5
\gamma$, and for every $i$, $bias(y_i) = 3 \gamma$, for some positive
real $\gamma$.
$E$ includes an edge from each $x_i$ to its corresponding $y_i$,
an edge in each direction between every $a$ neuron and every $y$
neuron, and a self-loop on each $y$ neuron.
Weights of the edges are as depicted in the figure.
The initial firing pattern $F_0$ gives arbitrary Boolean values for
the $a$ and $y$ neurons (technically, each $F_0$ yields a different network).
The initial values of the $x$ neurons are unspecified, indicating that
this network can be used with any inputs.


%

\subsection{Executions and probabilistic executions}

We describe how a network operates, beginning with its ordinary,
non-probabilistic executions and then adding probabilistic considerations.

\subsubsection{Executions and traces}



%
%

We begin by defining a ``configuration'' of a network, which describes
the current states of all neurons.
Namely, a \emph{configuration} of a neural network $\mathcal N$ is a firing
pattern for $N$, the set of all the neurons in the network.  
We consider several related definitions:   
\begin{itemize}
\item
An \emph{input configuration} is a firing pattern for the input
neurons, $N_{in}$.  
\item
An \emph{output configuration} is a firing pattern
for the output neurons, $N_{out}$.  
\item
An \emph{internal configuration} is
a firing pattern for the internal neurons, $N_{int}$.
\item
An \emph{external configuration} is a firing pattern for the input and
output neurons, $N_{ext}$.
\item
A \emph{non-input configuration} is a firing pattern for the internal
and output neurons, $N_{lc}$.
\end{itemize}

We define projections of configurations onto subsets of $N$. 
Thus, if $C$ is a configuration and $M$ is any subset of $N$, then $C
\lceil M$ is the firing pattern for $M$ obtained by projecting $C$
onto the neurons
in $M$.
In particular, we have $C \lceil N_{in}$ for the projection of $C$ on
the input neurons, 
$C \lceil N_{out}$ for the output neurons, 
$C \lceil N_{int}$ for the internal neurons, 
$C \lceil N_{ext}$ for the external neurons, 
and $C \lceil N_{lc}$ for the non-input neurons.
More generally, we can define the projection of any firing pattern $F$
for a set $M \subseteq N$ of neurons onto any subset $M' \subseteq M$.

An \emph{initial configuration} is a configuration $C$ such that 
$C \lceil N_{lc} = F_0$.
That is, the values for the locally-controlled neurons are as
specified by the given initial firing pattern.
The values for the input neurons are arbitrary.
We consider them to be controlled somehow, from outside the network.
For example, they may be output neurons of another network, or may
represent sensory inputs to the network.
%

Now we define formally how a network $\mathcal N$ executes; we assume
that it operates in synchronous rounds.
Namely, an \emph{execution} $\alpha$ of $\mathcal N$ is a (finite or infinite)
sequence of configurations, $C_0, C_1,\ldots,$ where $C_0$ is an initial
configuration.\footnote{We place no other restrictions on the general
  notion of an execution because our basic model does not impose any
  restriction on possible transitions.}
We define the \emph{length} of a finite execution $\alpha = C_0, C_1,...,C_t$,
$length(\alpha)$, to be $t$.  
As a special case, if $\alpha$ consists of just the initial
configuration $C_0$, then $length(\alpha) = 0$.
The \emph{length} of an infinite execution is defined to be $\infty$.

We define projections of executions onto subsets of the neurons of
$\mathcal N$. Namely, if $\alpha = C_0, C_1, \ldots$ is an execution
of $\mathcal N$ and $M$ is any subset of $N$, then $\alpha \lceil M$
is defined to be the sequence $C_0 \lceil M, C_1 \lceil M, \ldots$.
We define an $M$-execution of ${\mathcal N}$ to be $\alpha \lceil M$
for any execution $\alpha$ of ${\mathcal N}$.
%
We define an \emph{input execution} to be an $M$-execution where $M =
N_{in}$, and similarly for 
\emph{output execution}, 
\emph{internal execution}, 
\emph{external execution}, and
\emph{locally-controlled execution (or lc-execution)}.

To focus on the external behavior of the network, we define the notion of a 
``trace''.
Namely, for an execution $\alpha$, we write $trace(\alpha)$ as an
alternative notation for $\alpha \lceil N_{ext}$, the projection of
$\alpha$ on the external neurons.
We define a \emph{trace} of $\mathcal N$ to be the trace of any
execution $\alpha$ of $\mathcal N$.


\vspace{-.25cm}
\paragraph{Example:}
Again, consider the Winner-Take-All network.
Suppose that $F_0$, the initial firing pattern, assigns $0$ to all the
$a$ neurons and $y$ neurons, that is, none of these fire initially.
Then the executions of the network are just all the sequences of
configurations in which the starting configuration has values of $0$
for all the $a$ and $y$ neurons.
The values of the $x$ neurons are arbitrary.

\subsubsection{Probabilistic executions}
\label{sec: prob-execs}

We define a unique ``probabilistic execution'' for any particular
infinite input execution $\beta_{in}$.
First, we say that an infinite execution $\alpha$ of the network is
\emph{consistent with} $\beta_{in}$ provided that $\alpha \lceil N_{in}
= \beta_{in}$.
Also, a finite execution $\alpha$ is \emph{consistent with}
$\beta_{in}$ provided that $\alpha \lceil N_{in}$ is a prefix of
$\beta_{in}$.
Note that all of the (finite and infinite) executions that are
consistent with $\beta_{in}$ have the same initial configuration $C_0$.
This configuration is constructed from the first configuration
of $\beta_{in}$ and the initial non-input firing pattern for the
network, $F_0$.

The probabilistic execution for $\beta_{in}$ is defined as a probability
distribution $P$ on the sample space $\Omega$ of infinite executions
that are consistent with $\beta_{in}$.
The $\sigma$-algebra of measurable sets is generated from the
``cones'', each of which is the set of infinite executions in $\Omega$
that extend a particular finite execution.
Formally, if $\alpha$ is a finite execution that is consistent with
$\beta_{in}$, then $A(\alpha)$, the \emph{cone} of $\alpha$, is the set
of infinite executions that are consistent with $\beta_{in}$
and extend $\alpha$.
%
The other measurable sets in the $\sigma$-algebra are obtained by
starting with these cones and closing under countable union, countable
intersection, and complement.

\hide{
An equivalent definition is based on closure under countable
disjoint union, finite intersection, and complement.
In order to see that these definitions are equivalent, we argue that
closure under just countable disjoint union and finite intersection
and complement is enough to imply closure under countable union and
countable intersection:
\begin{enumerate}
\item
\emph{Closure under general finite union:}
For two sets, $A_1$ and $A_2$, we rewrite:
$A_1 \cup A_2$ as $A_1 \cup (A_2 - A_1) = A_1 \cup (A_2 \cap \not{A_1})$.
This last expression involves only disjoint union, finite
intersection, and complement.
The general case results from repeating this construction.
\item
\emph{Closure under general countable union:}
We express a union of the form
$A_1 \cup A_2 \cup A_3 \cup \ldots$ in disjoint form as
$A_1 \cup (A_2 - A_1) \cup (A_3 - (A_1 \cup A_2)) \cup \ldots$.
Each term of the form $A_i - (A_1 \cup \ldots A_{i-1})$ can be
rewritten as $A_i \cap \not(A_1 \cup \ldots A_{i-1})$
By 1. above, the term in parentheses can be rewritten using disjoint
union, finite intersection, and complement.  The rest of the
expression involves one more complement and finite intersection and
complement.
\item
\emph{Closure under countable intersection:}
We use DeMorgan's Law, based on general countable union and
complement.
\end{enumerate}
}

Now we define the probabilities for the measurable sets.
We start by explicitly defining the probabilities for the cones,
$P(A(\alpha))$.
Based on these, we can derive the probabilities of the other
measurable sets in a unique way, using general measure extension
theorems.
For example, Segala presents a similar construction for probabilistic
executions in his PhD thesis, Chapter 4~\cite{segalaThesis}.

We compute the probabilities $P(A(\alpha))$ recursively based on the
length of $\alpha$ (we assume here that $\alpha$ is consistent with
$\beta_{in}$):
\begin{enumerate}
\item
$\alpha$ is of length $0$. 

Then $\alpha$ consists of just the initial configuration $C_0$; define
$P(A(\alpha)) = 1$.

\item
$\alpha$ is of length $t$, $t > 0$.   

Let $\alpha'$ be the length-$(t-1)$ prefix of $\alpha$.
We determine the probability $q$ of extending $\alpha'$ to
$\alpha$.
Then the probability $P(A(\alpha))$ is simply $P(A(\alpha')) \times q$.

Let $C$ be the final configuration of $\alpha$ and $C'$ the final
configuration of $\alpha'$.
Then for each neuron $u \in N_{lc}$ separately, we use $C'$ and the
weights of $u$'s incoming edges to compute a potential and then
a firing probability for neuron $u$.
Specifically, for each $u$, we first calculate a \emph{potential},
$pot_u$, defined as 
\[
pot_u = \sum_{(v,u) \in E} C'(v) weight(v,u) - bias(u).
\]
We then convert $pot_u$ to a firing probability $p_u$ using a
standard sigmoid function:
\[
p_u = \frac{1}{1 + e^{-pot_u/\lambda}},
\]
where $\lambda$ is a positive real number ``temperature''
parameter.\footnote{
This function is called the sigmoid function because of its S-shape,
monotonically mapping the real line to the interval $[0,1]$.   
Although we assume a standard sigmoid function, the results of this
paper would also work with other S-shaped functions.}
We combine all those probabilities to compute the probability of generating
$C$ from $C'$:  for each $u \in N_{lc}$ such that $C(u) = 1$,
use the calculated probability $p_u$, and for each $u \in N_{lc}$ for
which $C(u) = 0$, use $1 - p_u$.
The product
\[
\prod_{u \in N_{lc} : C(u) = 1} p_u \times 
\prod_{u \in N_{lc} : C(u) = 0} (1-p_u)
\]
is the probability of generating $C$ from $C'$, which is the
probability $q$ of extending $\alpha'$ to $\alpha$.
\end{enumerate}

\vspace{-.25cm}
\paragraph{Example:}
Continuing with the Winner-Take-All network in Figure~\ref{fig:wta},
suppose again that $F_0$ assigns $0$ to all the non-input neurons.
Consider this network with the input configuration that assigns $1$ to 
$x_1$ and $0$ to all the other $x_i$ neurons.
Suppose that $\gamma = \lambda = 1$.
We compute the probability that $y_1$ fires.
The potential for neuron $y_1$ is $1 \times 3 - 3 = 0$, and the firing 
probability calculated from this using the standard sigmoid function is $.50$.
For any other $y$ neurons, we get potential $0 \times 3 - 3 = -3$, yielding
a firing probability of $.05$.

We will often consider conditional probabilities of the form
$P(A(\alpha_1) | A(\alpha_2))$.
Because we use a sigmoid function, we know that $P(A(\alpha_2))$
cannot be $0$, and so this conditional probability is well-defined.\footnote{
  One useful property of standard sigmoid functions is that the
  probabilities are never exactly $0$ or $1$, so we don't need to
  worry about $0$-probability sets when conditioning.}
%
The following lemma is straightforward.
\begin{lemma}
  \label{lem: execution-relative-probs}
Let $\alpha_1$ and $\alpha_2$ be finite executions of $\mathcal N$
that are consistent with $\beta_{in}$.
\begin{enumerate}
\item
If neither $\alpha_1$ nor $\alpha_2$ is an extension of the other,
that is, if they are incomparable,
then $P(A(\alpha_1) | A(\alpha_2)) = 0$.
\item
If $\alpha_1$  is an extension of $\alpha_2$,  then 
$P(A(\alpha_1) | A(\alpha_2)) = \frac{P(A(\alpha_1))}{P(A(\alpha_2))}$.
\end{enumerate}
\end{lemma}

Lemma~\ref{lem: execution-relative-probs} shows how we can compute the
conditional probabilities from the absolute probabilities.
Conversely, we can compute the absolute probabilities from the
conditional ones, as follows.

\begin{lemma}
Let $\alpha$ be a length-$t$ execution of $\mathcal N$, $t>0$, and 
suppose that $\alpha$ is consistent with $\beta_{in}$. 
Let $\alpha_i$, $0 \leq i \leq t$ be the successive prefixes of
$\alpha$ (so that $\alpha_0$ consists of the initial configuration
$C_0$ and $\alpha_t = \alpha$).
Then 
\[
P(A(\alpha)) = 
P(A(\alpha_1) | A(\alpha_0)) \times 
P(A(\alpha_2) | A(\alpha_1)) \cdots \times P(A(\alpha_t) | A(\alpha_{t-1})).
\]
\end{lemma}
Notice in the above expression, we did not start with a term for
$P(A(\alpha_0))$.  This is not needed because we are considering only
executions in which $\alpha_0$ is obtained from $\beta_{in}$ and the
initial assignment $F_0$.
So $P(A(\alpha_0)) = 1$.
Also note that each of the conditional terms is simply a one-step
transition probability, which can be calculated using the potential as
described above.

Since we can compute the conditional and absolute probabilities from
each other, either can be used to characterize the probabilistic execution.

\vspace{-.25cm}
\paragraph{Tree representation:}
%
%
%
%

The probabilistic execution for $\beta_{in}$ can be visualized as an
infinite tree of configurations, where the tree nodes at level $t$ 
represent the configurations that might occur at time $t$ (with the
given input execution $\beta_{in}$).
The configuration at the root of the tree is the initial configuration
$C_0$.
Each infinite branch of the tree represents an infinite execution of
the network, and finite initial portions of branches represent finite
executions.
Note that the same configuration can appear many times at different
vertices of the tree.

If $\alpha$ is a finite branch in the tree, then $P(A(\alpha))$ is the 
probability that an infinite execution will be in the ``cone'' of
executions that begin with $\alpha$.
We can associate the probability $P(A(\alpha))$ with the node at the
end of the finite branch---this is simply the probability of reaching
the node during probabilistic operation of the network, using the
inputs from $\beta_{in}$.

\subsubsection{Probabilities for projected executions}
\label{sec: properties-projected}

We extend the $A(\alpha)$ notation so that it applies to projections
of finite executions, not just complete finite executions.
Namely, suppose that $M$ is any subset of the neurons $N$ of $\mathcal
N$, and $\gamma$ is a finite $M$-execution of $\mathcal N$.
Then we say that $\gamma$ is \emph{consistent with} $\beta_{in}$
provided that $\gamma \lceil M \cap N_{in} = \beta_{in} \lceil M \cap
N_{in}$.
(This definition is equivalent to our earlier definition of
consistency in Section~\ref{sec: prob-execs}, for the special case
where $M = N$.)
In this case, we write $A(\gamma)$ for the set consisting of all
infinite executions $\alpha$ of $\mathcal N$ that are consistent with
$\beta_{in}$ such that $\gamma$ is a prefix of $\alpha \lceil M$.
%
%
We have:

\begin{lemma}
Let $M$ be any subset of the neurons $N$ of $\mathcal N$, and
let $\gamma$ be a finite $M$-execution of $\mathcal N$ that is
consistent with $\beta_{in}$.  Then, letting $\alpha$ range over the
set of finite executions that are consistent with $\beta_{in}$ and
such that $\alpha \lceil M = \gamma$:
\begin{enumerate}
\item
  \[
A(\gamma) = \bigcup_{\alpha} A(\alpha).
  \]
\item
  \[
P(A(\gamma)) = \sum_{\alpha} P(A(\alpha)).
  \]
\end{enumerate}
\end{lemma}


As an important special case, we consider $M = N_{ext}$, so that
$\gamma$ is specialized to a finite external execution $\beta$ of $\mathcal N$;
that is, we consider projections on the external neurons. 
Then our definition says that $\beta$ is \emph{consistent with} $\beta_{in}$
provided that $\beta \lceil N_{in} = \beta_{in}$.
In this case, we get:

\begin{lemma}
Let $\beta$ be a finite trace of $\mathcal N$ that is consistent with
$\beta_{in}$.  Then, letting $\alpha$ range over the
set of finite executions that are consistent with $\beta_{in}$ and
such that $trace(\alpha) = \beta$:
\begin{enumerate}
\item
  \[
A(\beta) = \bigcup_{\alpha} A(\alpha).
  \]
\item
  \[
P(A(\beta)) = \sum_{\alpha} P(A(\alpha)).
  \]
\end{enumerate}
\end{lemma}

We remark that the probabilities for finite executions and traces depend
only on their projections on the locally-controlled neurons, since the
input execution is always $\beta_{in}$.

\begin{lemma}
\label{lem: ignoring-inputs}
\begin{enumerate}
\item
Suppose that $\alpha$ is a finite execution of $\mathcal N$ that is
consistent with $\beta_{in}$.
Then $A(\alpha) = A(\alpha \lceil N_{lc})$ and $P(A(\alpha)) = P(A(\alpha \lceil N_{lc}))$.  
\item
Suppose that $\beta$ is a finite trace of $\mathcal N$ that is
consistent with $\beta_{in}$.
Then $A(\beta) = A(\beta \lceil N_{out})$ and $P(A(\beta)) = P(A(\beta \lceil N_{out}))$.  
\end{enumerate}
\end{lemma}

Now we give some simple lemmas involving the probabilities for finite
executions and related finite traces.
In the following lemma, the conditional probability statements follow
directly from the subset statements.

\begin{lemma}
  \label{lem: subsets}
Let $\alpha$ be a finite execution of $\mathcal N$ that
is consistent with $\beta_{in}$.
Suppose that $\alpha'$ is a prefix of $\alpha$.
Let $\beta = trace(\alpha) = \alpha \lceil N_{ext}$ and
$\beta' = trace(\alpha') = \alpha' \lceil N_{ext}$.
Then $\alpha'$, $\beta$, and $\beta'$ are also consistent with
$\beta_{in}$, and
\begin{enumerate}
\item
 $A(\alpha) \subseteq A(\beta)$, and
$P(A(\alpha)  | A(\beta)) = \frac{P(A(\alpha))}{P(A(\beta))}$.
\item
$A(\alpha) \subseteq A(\alpha')$, and 
$P(A(\alpha) | A(\alpha')) = \frac{P(A(\alpha))}{P(A(\alpha'))}$.
\item
$A(\alpha) \subseteq A(\beta')$, and 
$P(A(\alpha) | A(\beta')) = \frac{P(A(\alpha))}{P(A(\beta'))}$.
\item
$A(\alpha') \subseteq A(\beta')$, and
  $P(A(\alpha') | A(\beta')) = \frac{P(A(\alpha'))}{P(A(\beta'))}$.
\item
  $A(\beta) \subseteq A(\beta')$, and
  $P(A(\beta) | A(\beta')) = \frac{P(A(\beta))}{P(A(\beta'))}$.
\end{enumerate}
\end{lemma}

Consequences of the previous lemmas include the following, which is
used in Section~\ref{sec: compos-one-step}.

\begin{lemma}
  \label{lem: conditioning}
  \label{lem: claim3}
Let $\alpha$, $\alpha'$, $\beta$, and $\beta'$ be as in Lemma~\ref{lem: subsets}.
Then
\begin{enumerate}
\item
$P(A(\alpha) | A(\beta')) = P(A(\alpha) | A(\beta)) \times P(A(\beta) | A(\beta'))$.
\item
$P(A(\alpha) | A(\beta')) = P(A(\alpha) | A(\alpha')) \times P(A(\alpha') | A(\beta'))$.
\end{enumerate}
\end{lemma}

We also give a lemma about repeated conditioning, as for probabilistic executions:  

\begin{lemma}
\label{lem: extended-product-trace}
Let $\beta$ be a length-$t$ trace of $\mathcal N$, $t > 0$, and
suppose that $\beta$ is consistent with $\beta_{in}$.
Let $\beta_i$, $0 \leq i \leq t$, be the successive prefixes of
$\beta$ (so that $\beta_0$ consists of the initial configuration $C_0$
projected on $N_{ext}$ and $\beta_t = \beta$).
Then 
\[
P(A(\beta)) =
 P(A(\beta_1) | A(\beta_0)) \times 
P(A(\beta_2) | A(\beta_1) \cdots \times P(A(\beta_t) | A(\beta_{t-1})).
\]
\end{lemma}
As before, we do not need a separate term for $P(A(\beta_0))$,
because we are considering only traces in which $\beta_0$ is obtained
from $\beta_{in}$ and the initial assignment $F_0$.
So $P(A(\beta_0)) = 1$.

\subsubsection{Probabilistic traces}
\label{sec: prob-traces}

The previous definitions allow us to define a unique ``probabilistic
trace'' for any particular infinite input execution $\beta_{in}$.
The \emph{probabilistic trace} for $\beta_{in}$ is defined as a new
probability distribution $Q$, this one on the sample space $\Omega'$
of infinite traces $\beta$ that are consistent with $\beta_{in}$.
All of these traces have the same initial configuration,
constructed from the first configuration of $\beta_{in}$ and the initial
output firing pattern for the network, $F_0 \lceil N_{out}$.

The basic measurable sets are the sets of infinite traces in $\Omega'$
that extend a particular finite trace.
Formally, if $\beta$ is a particular finite trace that is consistent with
$\beta_{in}$, then $B(\beta)$, the ``cone'' of $\beta$, is the set of
infinite traces $\beta$ that are consistent with $\beta_{in}$ and
extend $\beta$.
Equivalently, $B(\beta)$ is just the set $traces(A(\beta))$.
Again, the other measurable sets in the $\sigma$-algebra are obtained
by starting with these cones and closing under countable union, countable
intersection, and complement.

We define the probabilities for the cones, $Q(B(\beta))$, based on
the corresponding probabilities for the probabilistic execution for
$\beta_{in}$.
Namely, if $\beta$ is a finite trace of $\mathcal N$ that is
consistent with $\beta_{in}$, then we define $Q(B(\beta))$ to be
simply $P(A(\beta))$.
As before, we can use these probabilities to derive the probabilities
of the other measurable sets in a unique way, using general measure
extension theorems as in~\cite{segalaThesis}.

\subsection{External behavior of a network}
\label{sec: extbeh-defs}

So far we have talked about individual probabilistic traces, each of
which depends on a fixed input execution $\beta_{in}$.
Now we define a notion of \emph{external behavior} of a network, which
is intended to capture its visible behavior for all possible inputs.
In Sections~\ref{sec: acyclic} and~\ref{sec: general}, we will show that
our notion of external behavior is \emph{compositional}, which means
that the external behavior of the composition of two networks,
${\mathcal N}^1 \times {\mathcal N}^2$,
is uniquely determined by the external behavior of ${\mathcal N}^1$
and the external behavior of ${\mathcal N}^2$.
          
Our definition of external behavior is based on the entire collection
of probabilities for the cones of all finite traces.
Namely, the external behavior $Beh(\mathcal N)$ is the mapping $f$
that maps each infinite input execution $\beta_{in}$ of $\mathcal N$
to the collection of probabilities $\{P(A(\beta))\}$ determined by the
probabilistic execution for $\beta_{in}$.
Here, $\beta$ ranges over the set of finite traces of $\mathcal N$ 
that are consistent with $\beta_{in}$.\footnote{
Formally, this ``collection'' is the mapping from finite traces
$\beta$ that are consistent with $\beta_{in}$ to the probabilities $P(A(\beta))$.
Thus, in terms of data types, $Beh(\mathcal N)$ is a nested mapping:
a mapping from the set of input executions to the set of mappings
from the set of finite traces consistent with $\beta_{in}$ to the set $[0,1]$.
}
In terms of probabilistic traces, this is the same as the collection
$\{Q(B(\beta))\}$, where $\beta$ has the same range.

\vspace{-.25cm}
\paragraph{Alternative behavior definitions:}
Other definitions of external behavior are possible.  Any such
definition would have to assign some ``behavior object'' to each
network $\mathcal N$.

In general, we define two external behavior notions $Beh_1$ and $Beh_2$ to
be \emph{equivalent} provided that the following holds.
Suppose that ${\mathcal N}$ and ${\mathcal N}'$ are two networks
with the same input neurons and the same output neurons.
Then $Beh_1(\mathcal N) = Beh_1({\mathcal N}')$ if and only if
$Beh_2(\mathcal N) = Beh_2({\mathcal N}')$.

%


Here we define one alternative behavior notion, based on one-step
conditional probabilities. 
This will be useful in our proofs for compositionality in
Section~\ref{sec: general}.
Namely, we define $Beh_2(\mathcal N)$ to be the mapping $f_2$
that maps each infinite input execution $\beta_{in}$ 
to the collection of conditional probabilities
$\{ P(A(\beta) | A(\beta')) \}$ based on the probabilistic execution
for $\beta_{in}$.
Here, $\beta$ ranges over the set of finite traces of $\mathcal N$
with length $> 0$ that are consistent with $\beta_{in}$, and $\beta'$
is the one-step prefix of $\beta$.

\begin{lemma}
The two behavior notions $Beh$ and $Beh_2$ are equivalent.
\end{lemma}

\begin{proof}
Suppose that ${\mathcal N}$ and ${\mathcal N}'$ are two networks
with the same input neurons and the same output neurons.
We show that $Beh$ and $Beh_2$ are equivalent by arguing 
the two directions separately:
\begin{enumerate}
\item
If $Beh({\mathcal N}) = Beh({\mathcal N}')$ then
$Beh_{2}({\mathcal N}) = Beh_{2}({\mathcal N}')$.

This follows because the conditional probability $P(A(\beta) | A(\beta'))$
is determined by the unconditional probabilities $P(A(\beta))$ and $P(A(\beta'))$;
see Lemma~\ref{lem: subsets}.

\item
If $Beh_{2}({\mathcal N}) = Beh_{2}({\mathcal N}')$ then
$Beh({\mathcal N}) = Beh({\mathcal N}')$.

This follows because the unconditional probability $P(A(\beta))$ is
determined by the conditional probabilities, see
Lemma~\ref{lem: extended-product-trace}.
\end{enumerate}
\qed
\end{proof}

\subsection{Examples}

In this subsection we give two fundamental examples to illustrate our
definitions so far:  some simple Boolean gate networks, and a network
implementing the ``Winner-Take-All'' mechanism from computational
neuroscience~\cite{trappenberg2009fundamentals,lazzaro1988winner,lee1999attention}.

\subsubsection{Simple Boolean gate networks}
\label{example: Boolean-circuits-1}

Figure~\ref{fig:boolean} depicts the structure of simple Spiking
Neural Networks in our model that represent and-gates, or-gates, and
not-gates.  For completeness, we also include an SNN representing the
identity computation.

\begin{figure}
\centering
    \begin{subfigure}{0.45\textwidth}
    \centering
        \includegraphics[width=.5\textwidth]{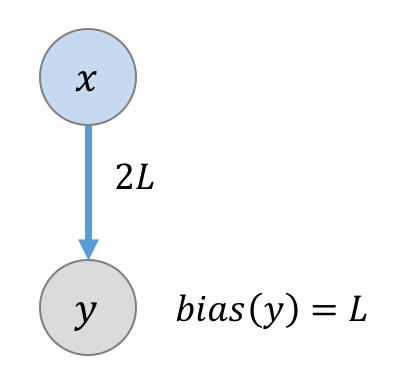}
        \caption{Identity}
    \end{subfigure}
    \begin{subfigure}{0.45\textwidth}
    \centering
        \includegraphics[width=.8\textwidth]{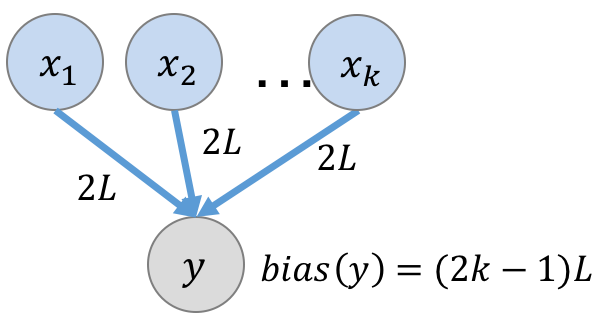}
        \caption{$k$-input And}
    \end{subfigure}
        \begin{subfigure}{0.45\textwidth}
        \centering
        \includegraphics[width=.7\textwidth]{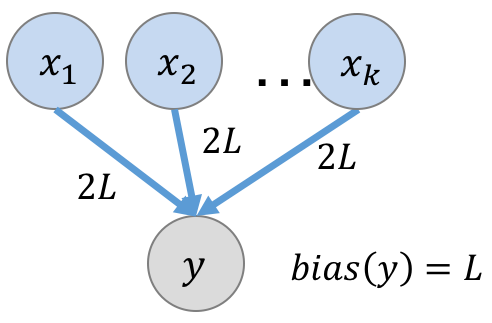}
        \caption{$k$-input Or}
    \end{subfigure}
    \begin{subfigure}{0.45\textwidth}
    \centering
        \includegraphics[width=.4\textwidth]{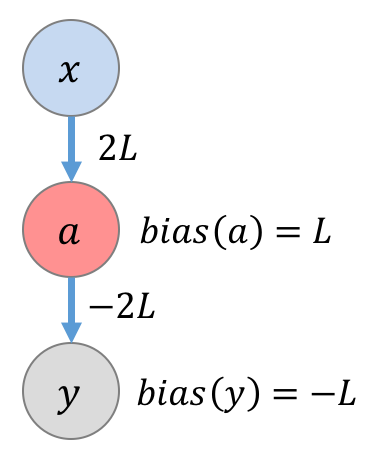}
        \caption{Not}
    \end{subfigure}
    \caption{Networks representing simple Boolean gates; here $L =
      \lambda \ln(\frac{1-\delta}{\delta})$, where $\delta$ is the
      error probability.
    }\label{fig:boolean}
\end{figure}

We describe the operation of each of these types of networks, in turn.
Fix a positive real number $\lambda$ for the temperature parameter of
the sigmoid function.
Fix an error probability $\delta$, $0 < \delta < 1$.
For each network below, let the initial firing pattern $F_0$ assign
$0$ to each locally controlled neuron.

Throughout this section, we use the abbreviation $L$ for the quantity
$\lambda \ln(\frac{1-\delta}{\delta})$; note that $L$ may be any real
number, but we focus on the case where $\delta \leq \frac{1}{2}$,
which makes $L$ non-negative.
We use the following identities repeatedly:
\[
e^{L/\lambda} = \frac{1-\delta}{\delta},
\frac{1}{1 + e^{L/\lambda}} = \delta, \mbox{ and }
\frac{1}{1 + e^{-L/\lambda}} = 1 - \delta.
\]

\vspace{-.25cm}
\paragraph{Identity network:}
The Identity network has one input neuron $x$ and one output neuron
$y$, connected by an edge with weight $w$.  The output neuron $y$ has
bias $b$.
Here we define $b = L$ and $w = 2L$.

With these settings, we get potential $w - b = 2L - L = L$ and
(expanding $L$, plugging into the sigmoid function, and using the
calculations above) output
firing probability $1-\delta$, in the case where the input fires.
Similarly, we get potential $-b = -L$ and output firing probability
$\delta$, in the case where the input does not fire.
Combining these two claims, consider the firing state of $x$ at time $0$.
Whether it is $0$ or $1$, the probability that $y$'s firing state
at time $1$ is the same as $x$'s firing state at time $0$ is exactly $1 - \delta$.

Now consider what happens with an arbitrary infinite input execution
$\beta_{in}$, rather than  just one input, that is, consider the
probabilistic execution for $\beta_{in}$.
Let $\beta$ be a finite trace of length $t \geq 1$ that is consistent with
$\beta_{in}$; by our assumption about $F_0$, $\beta$ must
include an initial firing state of $0$ for the output neuron $y$.
Suppose further that $\beta$ has the property that, for every $t'$, $1
\leq t' \leq t$, the firing state of $y$ at time $t'$ is equal to the
firing state of $x$ at time $t'-1$.  Then by repeated use of the
argument above, we get that $P(A(\beta)) = (1 - \delta)^t$.
      
Now suppose, as above, that $\beta$ is a length $t$ trace, $t \geq 1$,
that is consistent with $\beta_{in}$.
But now suppose that, in $\beta$,  the firing state of $y$ at time $t$ 
is equal to the firing state of $x$ at time $t-1$, but the firing
states of $y$ for all earlier times are arbitrary.
Let $\beta'$ denote the one-step prefix of $\beta$.
Then we can show that $P(A(\beta) | A(\beta')) = 1 - \delta$.
It follows that, for every time $t \geq 1$, the probability that the
firing state of $y$ at time $t$ is equal to the firing state of $x$ at
time $t-1$ is $1 - \delta$.
This uses the law of Total Probability, considering all the possible
length $t-1$ traces that are consistent with $\beta_{in}$. 

We also describe the external behavior $Beh$ for this network.
Namely, for each $\beta_{in}$, we must specify the collection of
probabilities $P(A(\beta))$, where $\beta$ ranges over the set of
finite traces of the network that are consistent with $\beta_{in}$.
In this case, for each such $\beta$ of length $t$, the probability
$P(A(\beta))$ is simply $(1-\delta)^a \delta^{t-a}$, where $a$ is the
number of positions $t'$, $1 \leq t' \leq t$, for which $y$'s firing
state in $\beta$ at time $t'$ is equal to $x$'s firing state in
$\beta$ at time $t'-1$.

\vspace{-.25cm}
\paragraph{$k$-input And network:}
The And network has $k$ input neurons, $x_1, x_2, \ldots, x_k$, and one output
neuron $y$.
Each input neuron is connected to the output neuron by an edge with
weight $w$.
The output neuron has bias $b$.
The Identity network is a special case of this network, where $k = 1$.

The idea here is to treat this as a threshold problem, and set $b$ and $w$
so that being over or under the threshold gives output firing state $1$ or $0$,
respectively, in each case with probability at least $1 - \delta$.
For a $k$-input And network, the output neuron $y$ should fire with
probability at least $1 - \delta$ if all $k$ input neurons fire, and
with probability at most $\delta$ if at most $k-1$ input neurons fire.

The settings for $b$ and $w$ generalize those for the Identity network.
Namely, define $b = (2k-1) L$ and $w = \frac{2b}{2k-1} =  2 L$.
When all $k$ input neurons fire, the potential is $k w - b = L$, 
and (expanding $L$ and plugging into the sigmoid function) the output firing
probability is $1 - \delta$.
When $k-1$ input neurons fire, the potential is $(k-1) w - b = -L$, 
and the output firing probability is $\delta$.
If fewer than $k-1$ fire, the potential and the output firing probability are
smaller. 
Similar claims about the external behavior $Beh$ for multi-round
computations to those we argued for the Identity network also hold for
the And network.

\vspace{-.25cm}
\paragraph{$k$-input Or network:}
The Or network has the same structure as the And network.  The
Or network also generalizes the Identity network, which is
the same as the $1$-input Or network.
Now the output neuron $y$ should fire with probability at least
$1 - \delta$ if at least one of the input neurons fires, and with
probability at most $\delta$ if no input neurons fire.
This time we set $b = L$ and $w = 2 L$. 
When one input neuron fires, the potential is 
$w - b = L$ and the output firing probability is $1 - \delta$.
When more than one fire, then the potential and the firing probability are greater.
When no input neurons fire, the potential is $-b = -L$, and the
output firing probability is $\delta$.
Again, similar claims about the external behavior for multi-round
computations hold for the Or network.

\vspace{-.25cm}
\paragraph{Not network:}
%
%
The Not network has one input $x$, one output $y$, and one internal
neuron $a$, which acts as an inhibitor for the output neuron.\footnote{We often
  classify neurons into two categories: \emph{excitatory neurons},
  all of whose outgoing edges have positive weights, and
  \emph{inhibitory neurons}, whose outgoing edges have negative
  weights.  However, this classification is not needed for the results
  in this paper.}
%
%
The network contains two edges, one from $x$ to $a$ with weight $w$,
and one from $a$ to $y$ with weight $w'$.
The internal neuron $a$ has bias $b$ and the output neuron $y$ has
bias $b'$.

The assembly consisting of the input and internal neurons acts like
the Identity network, with settings of $b$ and $w$ as before:
$b = L$ and $w = 2 L$.
So, for example, if we consider just $x$'s firing state at time
$0$, the probability that $a$'s firing state at time $1$ is
the same is exactly $1 - \delta$.

Let $b'$, the bias of the output neuron, be $-L$, and let $w'$, the
weight of the outgoing edge of the inhibitor, be $-2 L$.
Then if the internal neuron $a$ fires at time $1$, then the output neuron $y$
fires at time $2$ with probability $\delta$,
and if $a$ does not fire at time $1$, then $y$ fires at
time $2$ with probability $1-\delta$.
This yields probability $1- \delta$ of correct inhibition, which
then yields probabiity at least $(1 - \delta)^2$ that the output at
time $2$ gives the correct answer for the Not network.

Similar claims about multi-round computations as before also hold for
the Not network, except that the Not network has a delay of $2$
instead of $1$.
More precisely, consider an arbitrary infinite input execution
$\beta_{in}$, and consider the probabilistic execution for $\beta_{in}$.
Let $\beta$ be a finite trace of length $t \geq 2$ that is consistent
with $\beta_{in}$.
Then we know that $\beta$ must begin with a firing state of $0$ for
$y$; suppose also that the firing state of $y$ at time $1$ is $1$.  
Suppose further that $\beta$ has the property that, for every $t'$,
$2 \leq t' \leq t$, the firing state of $y$ at time $t'$ is unequal to the
firing state of $x$ at time $t'-2$.
Then we claim that $P(A(\beta)) \geq (1 - \delta)^{2(t-1)+1} = (1 - \delta)^{2t-1}$.
This is because, with probability $1-\delta$, the firing state of $y$
at time $1$ is equal to $1$,
and for each of the following times $t'$, $2 \leq t' \leq t$, with
probability at least $(1 - \delta)^2$, the firing state of $y$ at time
$t'$ is unequal to the firing state of $x$ at time $t'-2$.
%



\subsubsection{Winner-Take-All network}
\label{example: WTA-1}

Our next example is a simple \emph{Winner-Take-All (WTA)} network for
$n$ inputs and $n$ corresponding outputs.
It is based on a network presented in~\cite{LynchMP-arxiv19}.
Assume that some nonempty subset of the input neurons fire, in a
stable manner.
The output firing behavior is supposed to converge to a configuration
in which exactly one of the outputs, corresponding to one of the
firing inputs, fires.
We would like this convergence to occur quickly, in some fairly short
time $t_c$.  And we would like the resulting configuration to remain
stable for a fairly long time $t_s$.
Figure \ref{fig:wta} depicts the structure of the network.
There should be edges between every pair $(x_i,y_i)$ with weight
$3\gamma$, but these would be messy to draw.
    
\begin{figure}
\centering
\hspace{-1em}
 \includegraphics[width=.7\textwidth]{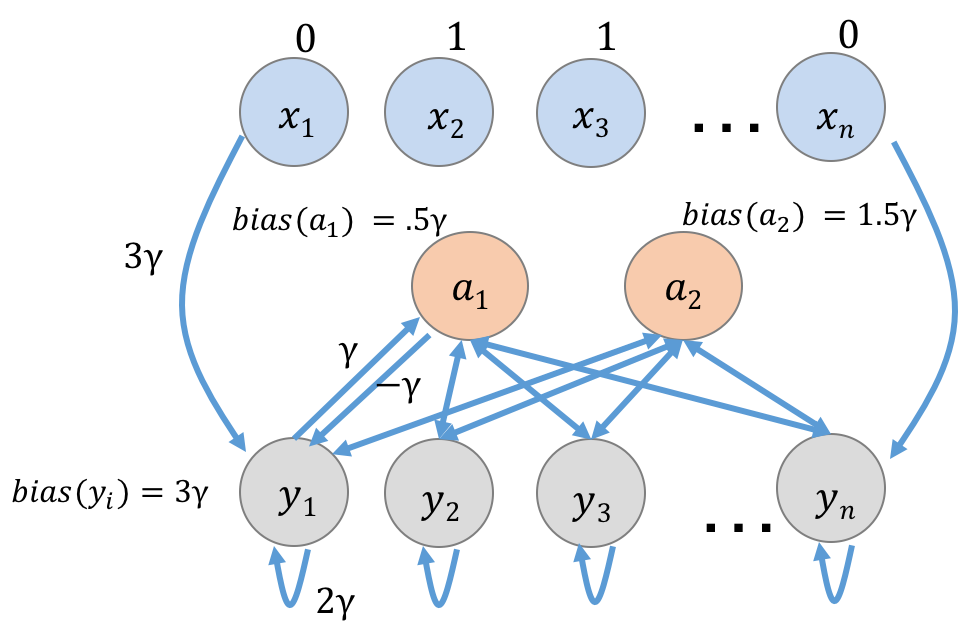}
\caption{A basic Winner-Take-All network.}
\label{fig:wta}
\end{figure}

In terms of the notation in this paper, consider any infinite input
execution $\beta_{in}$ in which all the input configurations are the
same and at least one input neuron is firing.
Consider the probabilistic execution for $\beta_{in}$.
In~\cite{LynchMP-arxiv19}, we prove that, in this probabilistic
execution, for certain values of $t_c$ and $t_s$, the probability
of convergence within time $t_c$ to an output configuration that
remains stable for time $t_s$ is at least $1 - \delta$.

The formal theorem statement is as follows.  Here, $\gamma$ is the
weighting factor used in the biases and edge weights in the network,
$\delta$ is a bound on the failure probability, and $c_1$ and $c_2$ are
particular small constants.

\begin{theorem}
\label{th: WTA}
Assume $\gamma \geq c_1 \log (\frac{n t_s}{\delta})$.
Then starting from any configuration, with probability $\geq 1 -
\delta$, the network converges, within time $t_c \le c_2 \log n
\log(\frac{1}{\delta})$, to a single firing output corresponding to a
firing input, and remains stable for time $t_s$. $c_1$ and $c_2$ are
universal constants, independent of $n$, $t_s$, and $\delta$.
\end{theorem}

In terms of our model, the desirable executions are determined by what
happens in their prefixes ending with time $t_c + t_s - 1$.
The correctness condition is that, within this prefix, there is a
consecutive sequence of $t_s$ times in which the output neurons
exhibit an unchanging firing pattern in which exactly one output $y_i$
fires, and we have $x_i = 1$ in the input configuration.
Note that this is a statement about external behavior (traces) only.
Correctness can be expressed formally in terms of the probabilities of
the cones starting with these desirable traces.

The proof appears in~\cite{LynchMP-arxiv19}.
The basic idea is that, when more than one output is firing, both
inhibitors are triggered to fire.  When they both fire, they cause
each firing output to continue firing with probability $\frac{1}{2}$.
This serves to reduce the number of firing outputs at a predictable
rate.  Once only a single output fires, only one inhibitor continues
to fire; its effect is sufficient to prevent other non-firing outputs
from beginning to fire, but not sufficient to stop the firing output
from firing.  All this, of course, is probabilistic.

Note that the network is symmetric with respect to the $n$ outputs.
Therefore, we can refine the theorem above to assert that, for any
particular output neuron $y_i$ that corresponds to a firing input
neuron $x_i$, the probability that $y_i$ is the eventual firing output
neuron is at least $\frac{1 - \delta}{n}$.

\section{Composition of Spiking Neural Networks}
\label{sec: composition}

In this section, we define composition of networks.
We focus on composing two networks, but the ideas extend in a
straightforward way to any finite number of networks.
Alternatively, we can describe multi-network composition by
repeated use of two-network composition.

\subsection{Composition of two networks}

Networks that are composed must satisfy some basic, natural compatibility
requirements.  These are analogous to those used for I/O automata and
similar models~\cite{Lynch-book,LynchT89,KLSV10}, except that instead of input and output
actions, we consider input and output neurons.
Namely, two networks ${\mathcal N}^1$ and ${\mathcal N}^2$ are said to
be \emph{compatible} provided that:
\begin{enumerate}
\item
No internal neuron of ${\mathcal N}^1$ is a neuron of ${\mathcal N}^2$.
\item
No internal neuron of ${\mathcal N}^2$ is a neuron of ${\mathcal N}^1$.
\item
No neuron is an output neuron of both ${\mathcal N}^1$ and ${\mathcal N}^2$.
\end{enumerate}
On the other hand, the two networks may have common input neurons, and
output neurons of one network may also be input neurons
of the other network.\footnote{
    In the brain setting, common input neurons for two different networks seem
    to make sense:  a neuron might have two different sets of outgoing
    edges (synapses), leading to different sets of neurons in the two networks.
}

\begin{lemma}
  \label{lem: no-common-edges}
  If $\mathcal N^1$ and $\mathcal N^2$ are compatible, then they do
  not have any edges in common.
\end{lemma}

\begin{proof}
Suppose for contradiction that they have a common edge, from a neuron
$u$ to a neuron $v$.  Then both $u$ and $v$ belong to both networks.
Since $v$ is shared, it must be an input neuron of at least one of the
networks, by compatibility.
But then that network has an edge leading to one of its input neurons,
which is forbidden by our network definition.
\qed
\end{proof}

Assuming ${\mathcal N}^1$ and ${\mathcal N}^2$ are compatible, we
define their composition 
$\mathcal N = {\mathcal N}^1 \times {\mathcal N}^2$ as follows:
\begin{itemize}
\item
$N$, the set of neurons of ${\mathcal N}$, is the union of $N^1$ and
$N^2$, which are the sets of neurons of ${\mathcal N}^1$ and
  ${\mathcal N}^2$ respectively.
Note that common neurons are included only once in the set $N$.

In network $\mathcal N$, each neuron retains its classification as 
input/output/internal from its sub-network, except that a neuron that
is an input of one sub-network and output of the other gets classified
as an output neuron of ${\mathcal N}$.
In particular, an output neuron of one sub-network that is also an
input neuron of the other sub-network remains an output neuron of $\mathcal
N$.\footnote{In Section~\ref{sec: hiding}, we will introduce a hiding
  operator that reclassifies some output neurons as internal neurons.}

Each non-input neuron in $\mathcal N$ inherits its $bias$ from its
original sub-network.
This definition of bias is unambiguous:  if a neuron belongs to both
sub-networks, it must be an input of at least one of them, and
input neurons do not have biases.

\item
  $E$, the set of edges of $\mathcal N$, is defined as follows.
If $e$ is an edge from neuron $u$ to neuron $v$ in either ${\mathcal N^1}$ or
${\mathcal N}^2$, then we include $e$ also in ${\mathcal N}$; these
are the only edges in $\mathcal N$.

Each edge inherits its weight from its original sub-network.
This definition of weight is unambiguous, by Lemma~\ref{lem: no-common-edges}. 

Thus, if the source neuron $u$ is an input of both sub-networks, then
in $\mathcal N$, $u$ has edges to all the nodes to which it has edges
in ${\mathcal N}^1$ and ${\mathcal N}^2$.
If $u$ is an output of one sub-network, say ${\mathcal N}^1$,  and an
input of the other, ${\mathcal N}^2$, then in ${\mathcal N}$, it has
all the incoming and outgoing edges it has in ${\mathcal N}^1$ as well
as the outgoing edges it has in ${\mathcal N}^2$.

On the other hand, the target neuron $v$ cannot be an input of both
networks since it has an incoming edge in one of them.
So $v$ must be an output of one, say ${\mathcal N}^1$, and an input of
the other, ${\mathcal N}^2$.  
Then in $\mathcal N$, $v$ has all the incoming and outgoing edges it
had in ${\mathcal N}^1$ as well as the outgoing edges it has in
${\mathcal N}^2$.

\item
$F_0$, the initial non-input firing pattern of $\mathcal N$, gets
inherited directly from the two sub-networks' initial
non-input firing patterns.  
Since the two sub-networks have no
non-input neurons in common, this is well-defined.
\end{itemize}

The probabilistic executions and probabilistic traces of the new
network ${\mathcal N}$ are defined in the usual way, as in
Section~\ref{sec: model}.
In Sections~\ref{sec: acyclic} and~\ref{sec: general}, we show how to
relate these to the probabilistic executions and probabilistic traces
of ${\mathcal N}^1$ and ${\mathcal N}^2$.

Here are some basic lemmas analogous to those in Section~\ref{sec:
  properties-projected}.
For these lemmas, fix $\mathcal N = {\mathcal N}^1 \times {\mathcal
  N}^2$ and a particular input execution $\beta_{in}$ of ${\mathcal
  N}$, which yields a particular probabilistic execution $P$.
Recall that we use the notation $N^j$ for the set of neurons of ${\mathcal N}^j$,
$j \in \{1,2\}$.

\begin{lemma}
\label{lem: exec-trace-j}
\label{lem: subsets1}
\label{lem: subsets2}
Let $\alpha$ be a finite execution of $\mathcal N$ 
that is consistent with $\beta_{in}$.
Suppose that $\alpha'$ is a prefix of $\alpha$.
Let $\beta = trace(\alpha) = \alpha \lceil N_{ext}$ and 
$\beta' = trace(\alpha') = \alpha' \lceil N_{ext}$.

Let $j \in \{1, 2\}$.
Let $\alpha^j = \alpha \lceil N^j$,
$\alpha'^j = \alpha' \lceil N^j$,  
$\beta^j = \beta \lceil N^j$, and
$\beta'^j = \beta' \lceil N^j$. 
Then $\alpha^j$, $\alpha'^j$, $\beta^j$, and $\beta'^j$ are also consistent with
$\beta_{in}$, and
\begin{enumerate}
\item
$A(\alpha^j) \subseteq A(\beta^j)$, and
$P(A(\alpha^j) | A(\beta^j)) = \frac{P(A(\alpha^j))}{P(A(\beta^j))}$.
\item
$A(\alpha^j) \subseteq A(\alpha'^j)$, and
$P(A(\alpha^j) | A(\alpha'^j)) = \frac{P(A(\alpha^j))}{P(A(\alpha'^j))}$.
\item
$A(\alpha^j) \subseteq A(\beta'^j)$, and
$P(A(\alpha^j) | A(\beta'^j)) = \frac{P(A(\alpha^j))}{P(A(\beta'^j))}$.
\item
$A(\alpha'^j) \subseteq A(\beta'^j)$, and
  $P(A(\alpha'^j) | A(\beta'^j)) = \frac{P(A(\alpha'^j))}{P(A(\beta'^j))}$.
\item
$A(\beta^j) \subseteq A(\beta'^j)$, and
$P(A(\beta^j) | A(\beta'^j)) = \frac{P(A(\beta^j))}{P(A(\beta'^j))}$.
\end{enumerate}
\end{lemma}

As before, the previous lemmas directly imply other properties, such as:

\begin{lemma}
\label{lem: conditioning1}
\label{lem: conditioning2}
Let $\alpha^j$, $\alpha'^j$, $\beta^j$, and $\beta'^j$ be as in
Lemma~\ref{lem: subsets1}.
Then
\begin{enumerate}
\item
$P(A(\alpha^j) | A(\beta'^j)) = P(A(\alpha^j) | A(\beta^j)) \times P(A(\beta^j)
  | A(\beta'^j))$.
%
\item
$P(A(\alpha^j) | A(\beta'^j)) = P(A(\alpha^j) | A(\alpha'^j)) \times
  P(A(\alpha'^j) | A(\beta'^j))$.
\end{enumerate}
\end{lemma}

Now we consider projections on the locally-controlled neurons of one
of the networks.
We have:

\begin{lemma}
\label{lem: claim4}
Let $\alpha$ be a finite execution of $\mathcal N$ that is consistent
with $\beta_{in}$.
Let $\alpha'$ be a prefix of $\alpha$ and $\beta' = trace(\alpha')$.
Let $j \in \{1,2\}$.
Then
\begin{enumerate}
\item
  $P(A(\alpha \lceil N^j_{lc}) | A(\alpha' \lceil N^j)) =
  \frac{P(A(\alpha \lceil N^j_{lc})  \cap A(\alpha' \lceil N^j)}{P(A(\alpha' \lceil N^j))}$.
\item
  $P(A(\alpha \lceil N^j_{lc}) | A(\beta' \lceil N^j)) =
  \frac{P(A(\alpha \lceil N^j_{lc})  \cap A(\beta' \lceil N^j)}{P(A(\beta' \lceil N^j))}$.
  \item
$P(A(\alpha \lceil N^j_{lc}) | A(\beta' \lceil N^j)) = 
 P(A(\alpha \lceil N^j_{lc}) | A(\alpha' \lceil N^j)) \times 
 P(A(\alpha' \lceil N^j) | A(\beta' \lceil N^j))$.
 \end{enumerate}
\end{lemma}

\begin{proof}
Parts 1 and 2 are just the definitions of conditional probability,
specialized to these sets.
For Part 3, note that
$A(\alpha \lceil N^j_{lc}) \cap  A(\beta' \lceil N^j) =
A(\alpha \lceil N^j_{lc}) \cap  A(\alpha' \lceil N^j)$,
because $\alpha \lceil N^j_{lc}$ already determines all the firing states
for neurons in $N^j_{lc}$.
Thus, we have that 
\[
P(A(\alpha \lceil N^j_{lc}) | A(\beta' \lceil N^j)) =
  \frac{P(A(\alpha \lceil N^j_{lc})  \cap A(\beta' \lceil N^j)}{P(A(\beta' \lceil N^j))}
\]
by Part 2, which is equal to 
\[
  \frac{P(A(\alpha \lceil N^j_{lc})  \cap A(\alpha' \lceil N^j))}{P(A(\beta' \lceil N^j))},
\]
which is in turn equal to 
\[
\frac{P(A(\alpha \lceil N^j_{lc}) \cap A(\alpha' \lceil N^j))}{P(A(\alpha' \lceil N^j))}
\times
     \frac{P(A(\alpha' \lceil N^j))}{P(A(\beta' \lceil N^j))}.
     \]
Part 1 and Lemma~\ref{lem: exec-trace-j} then imply that this is equal to
\[
P(A(\alpha \lceil N^j_{lc}) | A(\alpha' \lceil N^j)) \times 
P(A(\alpha' \lceil N^j) | A(\beta' \lceil N^j)),
\]
as needed.
\qed
\end{proof}

\vspace{-.25cm}
\paragraph{A special case:  acyclic composition:}
An important special case of composition is acyclic composition, in
which edges connect in only one direction, say from network ${\mathcal N}^1$
to network ${\mathcal N}^2$.
Formally, we say that a composition is \emph{acyclic} provided that it
satisfies the additional compatibility restriction
$N^1_{in} \cap N^2_{out} = \emptyset$, that is, output
  neurons of ${\mathcal  N}^2$ cannot be input neurons of ${\mathcal N}^1$.

Thus, ${\mathcal N}^1$ may have inputs only from the ``outside world'',
whereas its outputs can connect to ${\mathcal N}^1$, ${\mathcal N}^2$, and
the outside world.  
${\mathcal N}^2$ may have inputs from the outside world and from
${\mathcal N}^1$, and its outputs can connect only to ${\mathcal N}^2$ and the
outside world.

\subsection{Examples}

Here we give three examples.  The first two use acyclic
composition, and the third is a toy example that involves cycles.

\subsubsection{Boolean circuits}
\label{example: Boolean-circuits-2}

Figure \ref{fig:composition} contains a circuit that is a composition
of four Boolean gate circuits of the types described in
Section~\ref{example:  Boolean-circuits-1}:  two And networks,
one Or network, and a Not network.  We compose these networks into a
larger network that is intended to compute an Xor function.

\begin{figure}
\centering
 \includegraphics[width=1\textwidth]{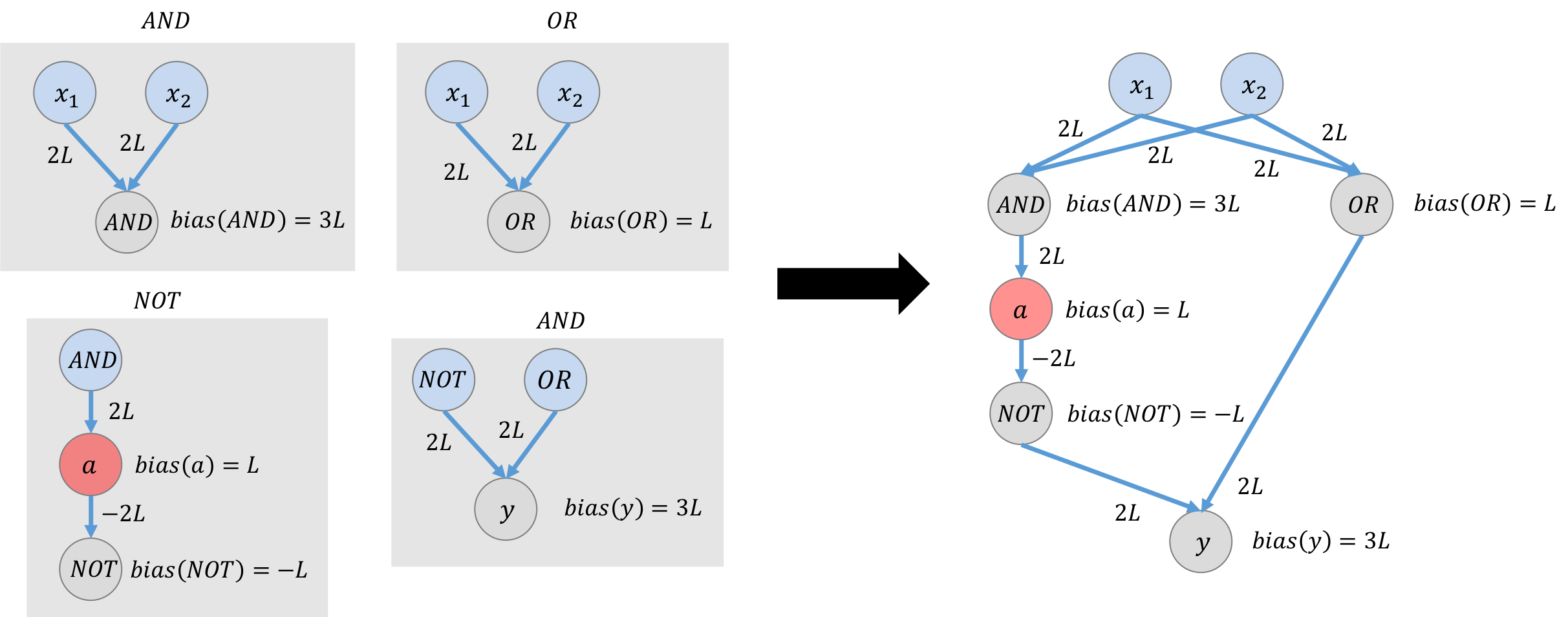}
\caption{Composing four Boolean gate circuits into an Xor network}
\label{fig:composition}
\end{figure}

In terms of the binary composition operator, we can compose the four
networks in three stages:
\begin{enumerate}
\item
Compose one of the And networks and the Not network to get a
network with two input neurons, two output neurons, and one internal
neuron, by identifying the output neuron of the And network with the
input neuron of the Not network.
Note that the composed network has two output neurons because the And
neuron remains an output---the composition operator does not
reclassify it as an internal neuron. 
The composed network is intended to compute the Nand of the two
inputs (as well as the And).
%
\item
Compose the network produced in Stage 1 with the Or network to get a
2-input-neuron, 3-output-neuron, 1-internal-neuron network, by identifying the 
corresponding inputs in the two networks.
The resulting network has output neurons corresponding to the Nand and the Or
of the two inputs (in addition to the And output neuron).
\item
Finally, compose the Nand network and the Or network with the second
And network, by identifying the Nand output neuron and the Or output
neuron with the two input neurons of the And network.
The resulting network has an output neuron corresponding to the Xor of the
two original inputs (in addition to outputs for the first And, the Nand,
and the Or networks).
\end{enumerate}
  
To state a simple guarantee for this composed circuit, let us assume
that the inputs fire consistently, in an unchanged firing pattern.
Then, working from the previously-shown guarantees of the individual
networks, we can say that the probability that the final output neuron $y$
produces its required Xor value at time $4$ is at least $(1 - \delta)^5$.
We revisit this example later, in Section~\ref{example: Boolean-circuits-3}.

\subsubsection{Attention using Winner-Take-All}
\label{example: WTA-2}

Figure \ref{fig:attention} depicts the composition of our
$WTA$ network from Section~\ref{example: WTA-1} with a
$2n$-input $n$ output $Filter$ network.  
The $Filter$ network is, in turn, a composition of $n$ disjoint And gates.
The composition is acyclic since information can flow from
$WTA$ to $Filter$ but not vice versa.  
  
\begin{figure}
\centering
 \includegraphics[width=1\textwidth]{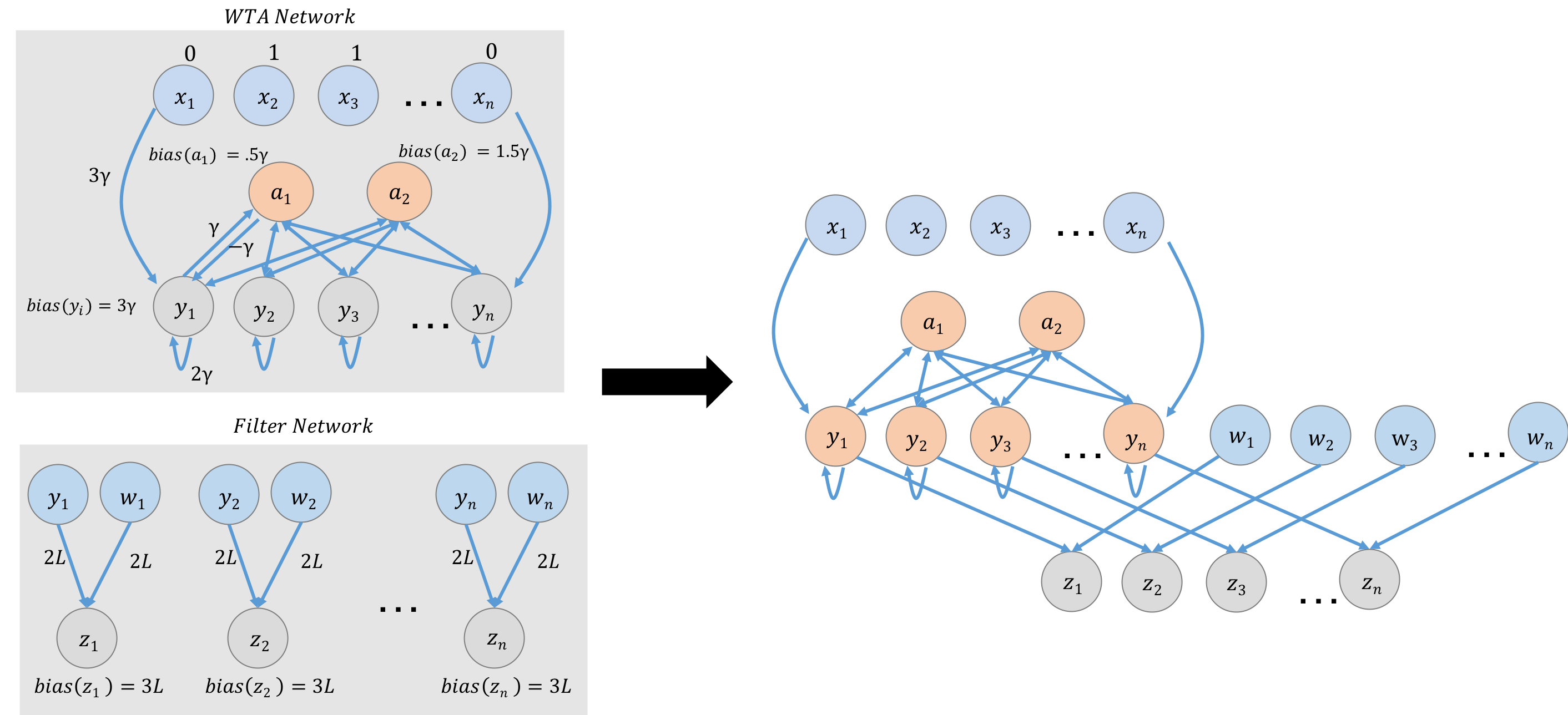}
\caption{An $Attention$ network built from a $WTA$ network and a $Filter$ network}
\label{fig:attention}
\end{figure}

The $Filter$ network is designed to fire any of its outputs, $z_i$, right
after the corresponding $w_i$ input fires, provided that its $y_i$
input (which is an output of the $WTA$ network) also fires.
In this way, the $WTA$ network is used to select particular outputs of
the $Filter$ network to fire---those that are ``reinforced'' by the
inputs from the $WTA$.

Assume that the $WTA$ and $Filter$ networks are composed, and the
$WTA$ inputs fire stably, with at least one input firing.
Then, as we described in Section~\ref{example: WTA-1},
with probability at least $1 - \delta$, the $WTA$ network soon stabilizes
to an output configuration with a single firing output $y_i$, which is
equally likely to be any of the $n$ outputs whose corresponding input
is firing.  That output
configuration should persist for a long time.
(Specific bounds are given in Theorem~\ref{th: WTA}.)

After the WTA stabilizes, it reinforces only a particular input $w_i$
for the $Filter$.  From that point on, the $Filter$'s $z_i$ outputs should
mirror its $w_i$ inputs, and no other $z$ outputs should fire.  
The probability of such mirroring should be at least 
$(1 - \delta')^{n t_s}$, if $\delta'$ denotes the failure probability
for an And gate.
(Recall from Example~\ref{example: WTA-1} that $t_s$ is the length of
the stable period for the $WTA$'s outputs.)
In this way, the composition can be viewed as an $Attention$
circuit, which pays attention to just a single input stream.

Note that the composed network behaves on two different time scales:
the $WTA$ takes some time to converge, but after that, the responses to
the selected intput stream will be essentially immediate.

\subsubsection{A toy example for cyclic composition}
\label{example: cyclic-1}

Now we give a toy example, consisting of two networks,
${\mathcal N}^1$ and ${\mathcal N}^2$,
that affect each other's behavior.
Throughout this section, we use the abbreviation $L$ for the quantity
$\lambda \ln(\frac{1-\delta}{\delta})$, as in
Section~\ref{example: Boolean-circuits-1}.
We assume that $\delta$ is ``sufficiently small''.

Figure \ref{fig:toy} shows a network ${\mathcal N}^1$ with one input
neuron $x_1$,
one output neuron $x_2$, and one internal neuron $a_1$. It has edges
from $x_1$ to $a_1$, from $a_1$ to $x_2$, and from $x_2$ to itself (a
self-loop).  The biases of $a_1$ and $x_2$ are $L$ and the weights on all
edges are $2L$.

\begin{figure}
\centering
 \includegraphics[width=1\textwidth]{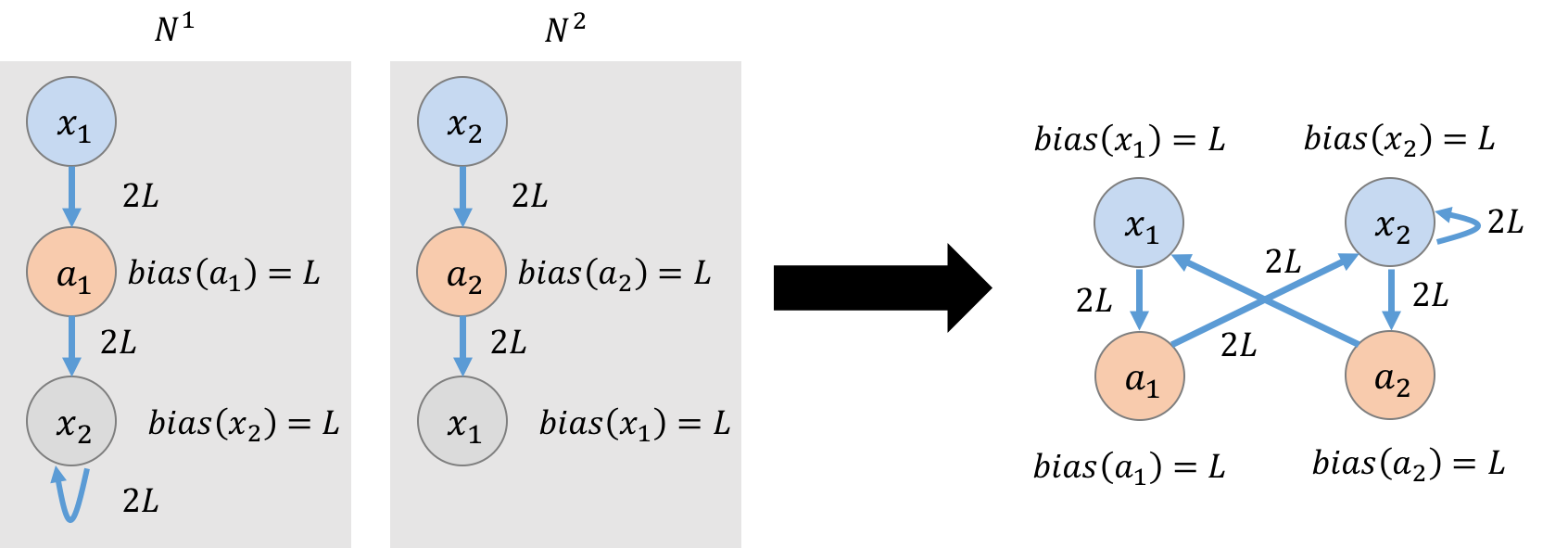}
\caption{A cyclic composition}
\label{fig:toy}
\end{figure}

Network ${\mathcal N}^1$ behaves so that, at any time $t \geq 1$, the firing
probability for the internal neuron $a_1$ is exactly $1 - \delta$ if
$x_1$ fires at time $t-1$, and is exactly $\delta$ if $x_1$ does not
fire at time $t-1$.
This is the same as for the output neuron of the Identity network in
Section~\ref{example: Boolean-circuits-1}.
The firing probability of the output neuron $x_2$ of ${\mathcal N}^1$
depends on the firing states of both $a_1$ and $x_2$ at time $t-1$.
This probability is:
\begin{itemize}
\item
$\delta$, if neither $a_1$ nor $x_2$ fires at time $t-1$.
\item
$1 - \delta$, if exactly one of $a_1$ and $x_2$ fires at time $t-1$.
\item
$1 - \frac{\delta^3}{(1 - \delta)^3 + \delta^3}$ if both $a_1$ and $x_2$
  fire at time $t-1$.
\end{itemize}
It follows that, if input $x_1$ fires at some time $t$, then output
$x_2$ is likely to fire at time $t+2$ (with probability at least
$(1 - \delta)^2$).
Without any additional input firing, and ignoring the low-likelihood
spurious firing of $a_1$,
the firing of $x_2$ is sustained only by the self-loop.
This means that the firing probability of $x_2$
decreases steadily over time, by a factor of $(1-\delta)$ at each
time.  Eventually, the firing should ``die out''.


Network ${\mathcal N}^2$ is similar, replacing $x_1$, $a_1$, and $x_2$
by $x_2$, $a_2$, and $x_1$, respectively.  However, we omit the
self-loop edge on $x_1$.  The biases are $L$ and the weights on the
two edges are $2L$.
Network ${\mathcal N}^2$ behaves so that, at any time $t \geq 1$, the firing
probability for the internal neuron $a_2$ is exactly $1 - \delta$ if
$x_2$ fires at time $t-1$, and is exactly $\delta$ if $x_2$ does not
fire at time $t-1$.
Likewise, the firing probability for the output neuron $x_1$ is
exactly $1 - \delta$ if $a_2$ fires at time $t-1$ and $\delta$ if
$a_2$ does not fire.
Thus, if input $x_2$ fires at some time $t$, then output 
$x_1$ is likely to fire at time $t+2$ (with probability at least $(1 -
\delta)^2$).
However, in this case, the firing of $x_1$ is not sustained.

Now consider the composition $\mathcal N = {\mathcal N}^1 \times {\mathcal N}^2$,
identifying the output $x_2$ of ${\mathcal N}^1$ with the input $x_2$
of ${\mathcal N}^2$, and the output $x_1$ of ${\mathcal N}^2$ with the
input $x_1$ of ${\mathcal N}^1$.
The behavior of $\mathcal N$ depends on the initial firing pattern.
Assume that neither $a_1$ nor $a_2$ fires initially; we
consider the behavior for the various starting firing patterns for
$x_1$ and $x_2$.
We consider two cases:
If neither $x_1$ nor $x_2$ fires at time $0$, then with ``high
probability'', none of the four neurons will fire for a long time.
On the other hand,
If one or both of $x_1$ and $x_2$ fire at time $0$, then with ``high
probability'', they will trigger all the neurons to fire and continue
to fire for a long time.
We give some details in Section~\ref{example: cyclic-2}.

\subsection{Compositionality definitions}

In Section~\ref{sec: extbeh-defs}, we defined a specific external
behavior notion $Beh$ for our networks, and an equivalent
alternative notion $Beh_2$.
Recall that, in general, a behavior definition $B$ assigns some
``behavior object'' $B(\mathcal N)$ to every network $\mathcal N$.
Here we define compositionality for general behavior
notions.  Later in the paper, in Sections~\ref{sec: acyclic}
and~\ref{sec: general}, we will prove that our particular behavior
notions are compositional.

In general, we define an external behavior notion $B$ to be
\emph{compositional} provided that the following holds:
Consider any four networks ${\mathcal N}^1$, ${\mathcal N}^2$, 
${\mathcal N}'^1$, and ${\mathcal N}'^2$, where 
${\mathcal N}^1$ and ${\mathcal N}'^1$ have the same sets of input and
output neurons, 
${\mathcal N}^2$ and ${\mathcal N}'^2$ have the same sets of input and
output neurons, 
${\mathcal N}^1$ and ${\mathcal N}^2$ are compatible, and
${\mathcal N}'^1$ and ${\mathcal N}'^2$ are compatible.
Suppose that $B({\mathcal N}^1) = B({\mathcal N}'^1)$ and
$B({\mathcal N}^2) = B({\mathcal N}'^2)$.
Then
$B({\mathcal N}^1 \times {\mathcal N}^2) = 
 B({\mathcal N}'^1 \times {\mathcal N}'^2)$.
Said another way:

\begin{lemma}
\label{lem: compositionality-characterization}
An external behavior notion $B$ is compositional if and only if,
for all compatible pairs of networks ${\mathcal N}^1$ and ${\mathcal
  N}^2$, $B({\mathcal N}^1 \times {\mathcal N}^2)$ is uniquely determined by 
$B({\mathcal N}^1)$ and $B({\mathcal N}^2)$.
\end{lemma}

Now we show that, in general, if two external behavior notions are
equivalent and one is compositional, then so is the other.
This will provide us with a method that will be helpful in
Section~\ref{sec: general} for showing compositionality.

\begin{theorem}
\label{th: composition-beh2}
If $B$ and $B'$ are two equivalent external behavior notions for
spiking neural networks, and $B$ is compositional, then also $B'$ is compositional.
\end{theorem}

\begin{proof}
Suppose that $B$ and $B'$ are two external behavior notions and
$B$ is compositional.  We show that $B'$ is compositional.
For this, consider any four networks ${\mathcal N}^1$, ${\mathcal N}^2$, 
${\mathcal N}'^1$, and ${\mathcal N}'^2$, where 
${\mathcal N}^1$ and ${\mathcal N}'^1$ have the same sets of input and
output neurons, 
${\mathcal N}^2$ and ${\mathcal N}'^2$ have the same sets of input and
output neurons, 
${\mathcal N}^1$ and ${\mathcal N}^2$ are compatible, and
${\mathcal N}'^1$ and ${\mathcal N}'^2$ are compatible.
Suppose that $B'({\mathcal N}^1) = B'({\mathcal N}'^1)$ and
$B'({\mathcal N}^2) = B'({\mathcal N}'^2)$.
We must show that
$B'({\mathcal N}^1 \times {\mathcal N}^2) = 
 B'({\mathcal N}'^1 \times {\mathcal N}'^2)$.

Since $B$ and $B'$ are equivalent and $B'({\mathcal N}^1) =
B'({\mathcal N}'^1)$, we have that $B({\mathcal N}^1) = B({\mathcal
  N}'^1)$.
Likewise, since $B'({\mathcal N}^2) = B'({\mathcal N}'^2)$, 
we have that $B({\mathcal N}^2) = B({\mathcal N}'^2)$.
Since $B$ is assumed to be compositional, this implies that
$B({\mathcal N}^1 \times {\mathcal N}^2) = 
B({\mathcal N}'^1 \times {\mathcal N}'^2)$. 
Then since $B$ and $B'$ are equivalent, we get that
$B'({\mathcal N}^1 \times {\mathcal N}^2) = 
B'({\mathcal N}'^1 \times {\mathcal N}'^2)$,
as needed.
\qed
\end{proof}

\section{Theorems for Acyclic Composition}
\label{sec: acyclic}

Our general composition results appear in Section~\ref{sec: general}.
Those are a bit complicated, mainly because of the possibility of
connections in both directions between the sub-networks.
Acyclic composition is an important special case of general
composition; many interesting examples satisfy the acyclic restriction.
Since this case can be analyzed more easily, we present this first.

Throughout this section, we fix the notation
$\mathcal N = {\mathcal N}^1 \times {\mathcal N}^2$, and assume that
$N^1_{in} \cap N^2_{out} = \emptyset$, that is, there are 
no edges from ${\mathcal N}^2$ to ${\mathcal N}^1$.

In this section, and from now on in the paper, we will generally avoid
writing the cone notation $A()$.
Thus, we will abbreviate $P(A(\alpha))$ and $P(A(\beta))$ as just
$P(\alpha)$ and $P(\beta)$.
We hope that this makes it easier to read complex formulas and does
not cause any confusion.

\subsection{Compositionality}
\label{sec: compos-acyclic}

We have not formally defined ``compositionality'' for the special case
of acyclic composition.
So here, 
we will simply show (Lemma~\ref{lem: acyclic-composition}) how
to express $Beh({\mathcal N})$ as a function of $Beh({\mathcal N}^1)$
and $Beh({\mathcal N}^2)$.  Thus (Theorem~\ref{th:  determine-Beh-acyclic}),
$Beh({\mathcal N})$ is uniquely determined by $Beh({\mathcal N}^1)$
and $Beh({\mathcal N}^2)$.

Specifically, we fix any particular input execution $\beta_{in}$ of
${\mathcal N}$, which generates a particular probability distribution
$P$ on infinite executions of ${\mathcal N}$.
We consider an arbitrary finite trace $\beta$ of ${\mathcal N}$ 
that is consistent with $\beta_{in}$.
We show how to express $P(\beta)$ in terms of probability
distributions $P^1$ and $P^2$ on infinite executions of
${\mathcal N}^1$ and ${\mathcal N}^2$, respectively.
These distributions $P^1$ and $P^2$ are defined from certain input
executions of ${\mathcal N}^1$ and ${\mathcal N}^2$, respectively.

We begin by deriving a simple expression for $P(\beta)$, for an
arbitrary finite trace $\beta$ of $\mathcal N$ that is consistent with
$\beta_{in}$, in terms of the same probability distribution $P$ on
projections of $\beta$.

\begin{lemma}
\label{lem: acyclic-composition-1}
Let $\beta$ be a finite trace of ${\mathcal N}$ that is consistent
with $\beta_{in}$.
Then
\[
P(\beta) = P(\beta \lceil N^1_{out}) \times 
           P((\beta \lceil N^2_{out}) | (\beta \lceil N^2_{in})).
\]
\end{lemma}

\begin{proof}
Since $\beta \lceil N_{in}$ is fixed, we have that 
\[
P(\beta) = P(\beta \lceil N_{out})
         = P((\beta \lceil (N^1_{out} \cup N^2_{out})).
\]
This last expression is equal to
\[
P(\beta \lceil N^1_{out}) \times 
P((\beta \lceil N^2_{out}) | (\beta \lceil N^1_{out}))
\]
by basic conditional probability reasoning.
We have that
\[
P((\beta \lceil N^2_{out}) | (\beta \lceil N^1_{out})) =
P((\beta \lceil N^2_{out}) | (\beta \lceil (N^1_{out} \cap N^2_{in}))),
\]
because the behavior of ${\mathcal N}^2$ does not depend on the firing
states of neurons in $N^1_{out} - N^2_{in}$.
(That is, the firing behavior of the neurons in $N^2_{out}$ is
independent of the behavior of the neurons in $N^1_{out} - N^2_{in}$,
conditioned on the behavior of the neurons in $N^1_{out} \cap N^2_{in}$.)
The right-hand side of this equation is equal to 
\[
P((\beta \lceil N^2_{out}) | (\beta \lceil (N^2_{in})))
\]
because $N^2_{in}$ consists of $N^1_{out} \cap N^2_{in}$ plus some
neurons in $N_{in}$, whose firing states are fixed in $\beta_{in}$.
Substituting yields
\[
P(\beta) = P(\beta \lceil N^1_{out}) \times 
           P((\beta \lceil N^2_{out}) | (\beta \lceil N^2_{in})),
\]
as needed.
\qed
\end{proof}

Thus, Lemma~\ref{lem: acyclic-composition-1} assumes an arbitrary input
execution $\beta_{in}$ of ${\mathcal N}$, which generates a
probability distribution $P$.
This lemma expresses $P(\beta)$, for an arbitrary $\beta$, in terms of
the $P$-probabilities of other finite traces.
However, we are not quite there:  Our main goal here is to express
$P(\beta)$ in terms of probability distributions $P^1$ and $P^2$ that
are generated by ${\mathcal N}^1$ and ${\mathcal N}^2$, respectively,
from particular infinite input executions for those respective sub-networks.  
We define these input executions and distributions as follows.
\begin{itemize}
\item
Input execution $\beta^1_{in}$ and distribution $P^1$ for ${\mathcal N}^1$:

Define the infinite input execution $\beta^1_{in}$ of ${\mathcal N}^1$
to be $\beta_{in} \lceil N^1_{in}$, that is, the projection of the
given input execution on the inputs of ${\mathcal N}^1$.
Then define $P^1$ to be the probability distribution that is generated by
${\mathcal N}^1$ from input execution $\beta^1_{in}$.

\item
Input execution $\beta^2_{in}$ and distribution $P^2$ for ${\mathcal N}^2$:

This is more complicated, since the input to ${\mathcal N}^2$
depends not only on the external input $\beta_{in}$, but also on the output
produced by ${\mathcal N}^1$.
Define the infinite input execution $\beta^2_{in}$ of ${\mathcal N}^2$
as follows.  
First, note that $N^2_{in} \subseteq N_{in} \cup N^1_{out}$, that is,
every input of ${\mathcal N}^2$ is either an input of ${\mathcal N}$
or an output of ${\mathcal N}^1$.
Define the firing patterns of the neurons in $N^2_{in} \cap N_{in}$ using
$\beta_{in}$, that is, define 
$\beta^2_{in} \lceil (N^2_{in} \cap N_{in}) = \beta_{in} \lceil N^2_{in}$.
And for the firing patterns of the neurons in $N^2_{in} \cap
N^1_{out}$, use $\beta$, that is, define
$\beta^2_{in} \lceil (N^2_{in} \cap N^1_{out}) = 
\beta \lceil (N^2_{in} \cap N^1_{out})$ for times
$0,\ldots,length(\beta)$ and the default $0$ for all later times.
(This choice for later times is arbitrary---we just chose $0$s to be concrete.)
Then define $P^2$ to be the probability distribution that is generated by
${\mathcal N}^2$ from input execution $\beta^2_{in}$.
\end{itemize}

Note that, in the second case above, the choice of the input execution
$\beta^2_{in}$ depends on the particular trace $\beta$ for which we are
trying to express the $P$-probability.
This is allowed because the external behavior $Beh({\mathcal N}^2)$
is defined to specify a probability distribution for \emph{every} individual infinite input
execution of ${\mathcal N}^2$.\footnote{To elaborate:  According to our
approach throughout this paper, we get a probability distribution of traces of ${\mathcal
  N}^2$ by fixing an infinite input execution of ${\mathcal N}^2$.
The question here is, which input to choose?
The infinite input $\beta_{in}$ for the entire system $\mathcal N$
provides part of the answer, for inputs of ${\mathcal N}^2$ that are
also inputs of $\mathcal N$.
The other part is obtained from $\beta$ projected on the inputs of
${\mathcal N}^2$ that are outputs of ${\mathcal N}^1$.
Technically, we have to pad out $\beta$ somehow, since we
need an infinite input execution, but it doesn't matter how we do
this, since the probability that ${\mathcal N}^2$ produces outputs
consistent with $\beta$ depends only on the portion of the input
up to $length(\beta)$.}

The next lemma restates the result of Lemma~\ref{lem:
  acyclic-composition-1} in terms of the new probability distributions
$P^1$ and $P^2$.
The key idea is that the probability $P^2$ is essentially a
conditional probability distribution, giving
probabilities for ${\mathcal N}^2$'s outputs, conditioned on its
inputs being consistent with $\beta$.

\begin{lemma}
\label{lem: acyclic-composition-2}
Let $\beta$ be a finite trace of $\mathcal N$ that is consistent with
$\beta_{in}$.  Then 
\[
P(\beta) =
P^1(\beta \lceil N^1_{out}) \times P^2(\beta \lceil N^2_{out}).
\]
\end{lemma}

\begin{proof}
Fix $\beta$, a finite trace of of ${\mathcal N}$ that is consistent
with $\beta_{in}$.
By Lemma~\ref{lem: acyclic-composition-1}, we know that:
\[
P(\beta) = P(\beta \lceil N^1_{out}) \times 
P((\beta \lceil N^2_{out}) | (\beta \lceil N^2_{in})).
\]
It suffices to show that these two terms are equal to the
corresponding terms in this lemma, that is, that
\[
P(\beta \lceil N^1_{out})  = P^1(\beta \lceil N^1_{out}) 
\]
and 
\[
P((\beta \lceil N^2_{out}) | (\beta \lceil N^2_{in})) = P^2(\beta \lceil N^2_{out}).
\]
These two statements follow directly by unwinding the definitions of
$P^1$ and $P^2$, respectively.
%
Specifically, for the first statement, we consider $P(\beta \lceil
N^1_{out})$, the probability that the composed network $\mathcal N$
generates an execution that, when projected on outputs of ${\mathcal
  N}^1$, starts with $\beta \lceil N^1_{out}$.
We note that this probability is entirely determined by the sub-network
${\mathcal N}^1$, based on $\beta_{in}$ projected on the
inputs of ${\mathcal N}^1$.
But this is just the definition of $P^1(\beta \lceil N^1_{out})$.

Likewise, though a bit more subtly, for the second statement, we consider
$P((\beta \lceil N^2_{out}) | (\beta \lceil N^2_{in}))$, which is the conditional
probability that the composed network generates an execution that,
when projected on outputs of ${\mathcal N}^2$, starts with $\beta
\lceil N^2_{out}$, conditioned on the event that the inputs to
${\mathcal N}^2$ start with $\beta \lceil N^2_{in}$.
This time, the probability is entirely determined by the sub-network
${\mathcal N}^2$, based on $\beta$ projected on the
inputs of ${\mathcal N}^2$.\footnote{Notice that this
  probability is entirely determined by the finite input $\beta \lceil
  N^2_{in}$---the firing states of the input neurons of
  ${\mathcal N}^2$ after time $length(\beta)$ do not matter.}
But this is just the definition of $P^2(\beta \lceil N^2_{out})$.
\qed
\end{proof}
%

The next lemma has a slightly simpler statement than Lemma~\ref{lem:
  acyclic-composition-2}.

\begin{lemma}
\label{lem: acyclic-composition}
Let $\beta$ be a finite trace of $\mathcal N$ that is consistent with
$\beta_{in}$.
Then 
\[
P(\beta) = P^1(\beta \lceil N^1) \times P^2(\beta \lceil N^2).
\]
\end{lemma}

\begin{proof}
This follows from Lemma~\ref{lem: acyclic-composition-2} because
in each term on the right-hand-side of the equation in this lemma, the
probability depends on the output traces only---the input traces are fixed.  
Formally, this uses Lemma~\ref{lem: ignoring-inputs}.
\qed
\end{proof}

Finally, Lemma~\ref{lem: acyclic-composition} yields a kind of
compositionality theorem for acyclic composition:

\begin{theorem}
\label{th:  determine-Beh-acyclic}
$Beh({\mathcal N})$ is determined by $Beh({\mathcal N}^1)$ and
$Beh({\mathcal N}^2)$.
\end{theorem}

We prove a more general compositionality result in Section~\ref{sec:
  general}.

\subsection{Examples}

We revisit our two examples of acyclic composition from
Sections~\ref{example: Boolean-circuits-2} and~\ref{example: WTA-2},
this time analyzing their behavior more precisely.

\subsubsection{Boolean circuits}
\label{example: Boolean-circuits-3}

Let ${\mathcal N}$ be the seven-neuron Boolean circuit from
  Section~\ref{example: Boolean-circuits-2}.
Express $\mathcal N$ as the composition ${\mathcal N}^1 \times
{\mathcal N}^2$, where
\begin{itemize}
  \item
${\mathcal N}^1$ is the network resulting from the first two
stages in the order of compositions described in Section~\ref{example:
  Boolean-circuits-2}.  This computes Nand and Or of the two inputs.
\item
${\mathcal N}^2$ is the final And network.
\end{itemize}

Fix $\beta_{in}$ to be any infinite input execution of $\mathcal N$
with stable inputs, and let $P$ be the probabilistic execution of
$\mathcal N$ for $\beta_{in}$.
In $P$, we should expect to have stable, correct outputs for a long
while starting from time $4$, because the depth of the entire network
is $4$.
Here we consider just the situation at precisely time $4$, that is, we
consider the probabilities $P(\beta)$ for finite traces $\beta$ of
length exactly $4$.
Specifically, we would like to use Lemma~\ref{lem:
  acyclic-composition-2} to help us
show that the probability of a correct Xor output at time $4$ is at
least $(1 - \delta)^5$.

We work compositionally.
In particular, we assume that, in the probabilistic execution of
${\mathcal N}^1$ for $\beta_{in}$, or any other stable input sequence,
the probability of correct (Nand,Or) outputs at time $3$ is at least
$(1 - \delta)^4$.
%
We also assume that, in the probabilistic execution of ${\mathcal N}^2$ 
on any input sequence, the probability that the output at time $4$ is
the And of its two inputs at time $3$ is at least $1 - \delta$.
We could prove these bounds for our two specific networks
${\mathcal N}^1$ and ${\mathcal N}^2$, but to emphasize the
compositional reasoning, we ignore the internal workings of the two
sub-networks and simply state the bounds here.  
We use these bounds to get our result about the composed network
$\mathcal N$.

So define $B$ to be the set of traces $\beta$ of $\mathcal N$ of
length $4$ such that $\beta$ gives a correct Xor output at time $4$,
as well as correct (Nand, Or) outputs at time $3$.
(These traces may differ in their firing states for the And neuron at
any time, and also in their firing states for the Not and Or neurons
at times other than those specified.)
We will argue that $P(B) \geq (1 - \delta)^5$, which implies our
desired result.

We have that $P(B) = \sum_{\beta \in B} P(\beta)$.
By Lemma~\ref{lem: acyclic-composition-2}, this is equal to
\[
\sum_{\beta \in B} P^1(\beta \lceil N^1_{out}) \times P^2(\beta \lceil N^2_{out}).
\]
Here, $P^1$ and $P^2$ are defined as in Section~\ref{sec: compos-acyclic}, 
based on $\beta^1_{in} = \beta_{in}$,
and for each particular $\beta$, based on $\beta^2_{in}$ equal to
$\beta \lceil N^2_{in}$, extended to an infinite sequence by adding $0$'s.
Note that the choice of input sequence $\beta^2_{in}$ for ${\mathcal
  N}^2$ is uniquely determined by $\beta \lceil N^1_{out}$.

We break this expression up into the double summation: 
\[
\sum_{\beta^1} 
(\sum_{\beta^2}  
P^1(\beta^1 \lceil N^1_{out}) \times P^2(\beta^2 \lceil N^2_{out}))
\]
Here, $\beta^1$ ranges over traces of ${\mathcal N}^1$ that are
consistent with $\beta_{in}$ and yield correct (Nand, Or) outputs at
time $3$.
And for each particular $\beta^1$, $\beta^2$ ranges over traces of
${\mathcal N}^2$ that are consistent with the input sequence
$\beta^2_{in}$ determined from $\beta^1 \lceil N^1_{out} = \beta
\lceil N^1_{out}$, and whose output at time $4$ is the Xor of its
inputs at time $3$.
This is equal to (collecting terms for each $\beta^1$):
\[
\sum_{\beta^1} 
P^1(\beta^1 \lceil N^1_{out}) \sum_{\beta^2} P^2(\beta^2 \lceil N^2_{out}).
\]

Now, for any particular $\beta^1$, we know that:
\[
\sum_{\beta^2} P^2(\beta^2 \lceil N^2_{out}) \geq (1 - \delta),
\]
by our assumptions about the behavior of ${\mathcal N}^2$. 
So the overall expression is at least
\[
\sum_{\beta^1} P^1(\beta^1 \lceil N^1_{out}) (1 - \delta)
=
(1 - \delta) \sum_{\beta^1} P^1(\beta^1 \lceil N^1_{out}).
\]
We also know that 
\[
\sum_{\beta^1} P^1(\beta^1 \lceil N^1_{out}) \geq (1 - \delta)^4,
\]
by our assumption about the behavior of ${\mathcal N}^1$. 
So the overall expression is at least
$(1 - \delta) (1 - \delta)^4 = (1 - \delta)^5$, as needed.


\subsubsection{Attention using WTA}
\label{example: WTA-3}

We consider the composition of the $WTA$ network and the $Filter$ network,
as described in Section~\ref{example: WTA-2}.
Now let ${\mathcal N}^1$ denote the $WTA$ network, ${\mathcal N}^2$
the $Filter$ network, and  $\mathcal N$ their composition.
We assume that the $WTA$ network satisfies Theorem~\ref{th: WTA}, with
particular values of $\delta$, $t_c$, $t_s$, $\gamma$, $c_1$ and $c_2$.
We assume that each And network within $Filter$ is correct at each time
with probability at least $1 - \delta'$.

Fix $\beta_{in}$ to be any infinite input execution of $\mathcal N$
with stable $x_i$ inputs such that at least one $x_i$ is firing.
The $w_i$ inputs are unconstrained.
Let $P$ be the probabilistic execution of $\mathcal N$ generated from
$\beta_{in}$.
We want to prove that, according to $P$, with probability at least 
$(1 - \delta) (1 - \delta')^{n t_s}$, there is some $t \leq t_c$
such that:
(a) the $y$ outputs stabilize by time $t$ to one steadily-firing
output $y_i$, which persists through time $t + t_s - 1$, and
(b) for this particular $i$, starting from time $t + 1$ and continuing
for a total of $t_s$ times, the $z_i$ outputs correctly mirror the
$w_i$ inputs at the previous time, and all the other $z$ neurons
do not fire.

Again, we work compositionally.
We assume that, in the probabilistic execution of the WTA network
${\mathcal N}^1$ on $\beta_{in} \lceil N_{in}$, the probability of
correct, stable outputs as in Theorem~\ref{th: WTA} is at least $1 - \delta$.
We also assume that, in the probabilistic execution of ${\mathcal
  N}^2$ on any input sequence, conditioned on any finite execution
prefix, the probability of correct mirroring of inputs for the
next $t$ times is at least $(1 - \delta')^{n t_s}$.
These assumptions could be proved for our two networks, but we simply
assume them here.

Now define $B$ to be the set of traces $\beta$ of ${\mathcal N}$ of
length $t_c + t_s - 1$ such that all the desired conditions hold in
$\beta$, that is, there is some $t \leq t_c$ such that in $\beta$,
(a) the $y$ outputs stabilize by time $t$ to one steadily-firing
output $y_i$, which persists through time $t + t_s - 1$, and
(b) for this particular $i$, starting from time $t + 1$ and continuing
for a total of $t_s$ times, the $z_i$ outputs correctly mirror the
$w_i$ inputs at the previous time, and all the other $z$ neurons
do not fire.
We will argue that $P(B) \geq (1 - \delta) (1 - \delta')^{n t_s}$.
We follow the same pattern as in the Boolean circuit network example
in Section~\ref{example: Boolean-circuits-3}.

We have that $P(B) = \sum_{\beta \in B} P(\beta)$.
By Lemma~\ref{lem: acyclic-composition-2}, this is equal to
\[
\sum_{\beta \in B} P^1(\beta \lceil N^1_{out}) \times P^2(\beta \lceil N^2_{out}).
\]
Here, $P^1$ and $P^2$ are defined as in Section~\ref{sec: compos-acyclic}, 
based on $\beta^1_{in} = \beta_{in} \lceil N^1_{in}$ 
and for each particular $\beta$, based on $\beta^2_{in}$ equal to
$\beta \lceil N^2_{in}$, extended to an infinite sequence by adding $0$'s.
Note that $\beta^2_{in}$ is uniquely determined by 
$\beta \lceil (N_{in} \cup N^1_{out})$.

This expression is equal to:
\[
\sum_{\beta^1} 
(\sum_{\beta^2}  
P^1(\beta^1 \lceil N^1_{out}) \times P^2(\beta^2 \lceil N^2_{out})).
\]
Here, $\beta^1$ ranges over traces of ${\mathcal N}^1$ that are
consistent with $\beta_{in}$ and for which there is some $t \leq t_c$
such that in $\beta^1$, the $y$ outputs stabilize by time $t$ to one
steadily-firing output $y_i$, which persists through time $t + t_s -
1$.
And for each particular $\beta^1$, $\beta^2$ ranges over traces of
${\mathcal N}^2$ that are consistent with the input sequence
$\beta^2_{in}$ determined from $\beta_{in}$ and 
$\beta^1 \lceil N^1_{out} = \beta \lceil N^1_{out}$, and that satisfy
the following correctness condition for ${\mathcal N}^2$:  for the first
$t$ and associated $i$ that witness the correctness condition for $\beta^1$,
at times $t+1,\ldots,t + t_s$, the $z_i$ outputs correctly mirror the
$w_i$ inputs at the previous time, and all the other $z$ neurons
do not fire.

This is equal to (collecting terms for each $\beta^1$):
\[
\sum_{\beta^1} 
P^1(\beta^1 \lceil N^1_{out}) \sum_{\beta^2} P^2(\beta^2 \lceil N^2_{out}).
\]

Now, for any particular $\beta^1$, we know that:
\[
\sum_{\beta^2} P^2(\beta^2 \lceil N^2_{out}) \geq (1 - \delta')^{n t_s},
\]
by our assumptions about the behavior of ${\mathcal N}^2$. 
So the overall expression is at least
\[
\sum_{\beta^1} P^1(\beta^1 \lceil N^1_{out}) (1 - \delta')^{n t_s}
=
(1 - \delta')^{n t_s} \sum_{\beta^1} P^1(\beta^1 \lceil N^1_{out}).
\]
We also know that 
\[
\sum_{\beta^1} P^1(\beta^1 \lceil N^1_{out}) \geq (1 - \delta),
\]
by our assumption about the behavior of ${\mathcal N}^1$. 
So the overall expression is at least
$(1 - \delta) (1 - \delta')^{n t_s}$, as needed.

\section{Theorems for General Composition}
\label{sec: general}

For general composition, the simple approach  in Section~\ref{sec:
  acyclic} does not work.
There, we were able to prove results such as Lemma~\ref{lem:
  acyclic-composition-1},
which decompose the behavior of the entire network $\mathcal N$ in
terms of the behavior of the two sub-networks ${\mathcal N}^1$ and
${\mathcal N}^2$.
This worked because the dependencies between the behaviors go only one
way, from ${\mathcal N}^1$ to ${\mathcal N}^2$.
In the general case, the dependencies go both ways, potentially leading to
circularities.

Fortunately, since we are working in a synchronous model, we
can break the circularities in another way, using discrete time.
Namely, the behavior of each sub-network at time $t$ depends only on
the behavior of the other network at times up to $t-1$.
We exploit this limitation on dependencies to prove decomposition
lemmas such as Lemma~\ref{lem: independence-executions}, leading to
our main compositionality theorem, Theorem~\ref{th: beh-compositional}.

For this section, fix 
$\mathcal N = {\mathcal N}^1 \times {\mathcal N}^2$.
We continue to avoid writing the cone notation $A()$.

\subsection{Composition results for executions and traces}

For this subsection and the following, fix a particular input
execution $\beta_{in}$ for ${\mathcal N}$, which yields a particular
probabilistic execution $P$.
The main result of this subsection is Lemma~\ref{lem:
  independence-executions}.
It says that the probability of a certain finite execution $\alpha$ of the
entire network $\mathcal N$, conditioned on its trace $\beta$, is
simply the product of the probabilities of the two projections of
$\alpha$ on the two sub-networks, each conditioned on its projected trace.
In other words, once we fix all the external behavior of the network,
including the part of the behavior involved in interaction between the
two sub-networks, the internal states of the neurons within the two
sub-networks are determined independently.  
We begin with a straightforward lemma that treats the two sub-networks
asymmetrically.

\begin{lemma}
  \label{lem: independence-executions-1}
Let $\alpha$ be a finite execution of $\mathcal N$ that is consistent
with $\beta_{in}$, and let $\beta = trace(\alpha)$.
Then 
\[
P(\alpha | \beta) = 
P((\alpha \lceil N^1_{int}) | \beta) \times
P((\alpha \lceil N^2_{int}) | (\alpha \lceil N^1_{int}), \beta).
\]
\end{lemma}

\begin{proof}
Standard conditional probability.
\qed
\end{proof}

And now we remove the asymmetry, by identifying the portions of
$\beta$ on which the internal behavior of the two sub-networks
actually depends.

\begin{lemma}
\label{lem: independence-executions}
Let $\alpha$ be a finite execution of $\mathcal N$ that is consistent
with $\beta_{in}$, and let $\beta = trace(\alpha)$.
Then 
\[
P(\alpha | \beta) = 
P((\alpha \lceil N^1) | (\beta \lceil N^1)) \times
P((\alpha \lceil N^2) | (\beta \lceil N^2)).
\]
\end{lemma}

\begin{proof}
Lemma~\ref{lem: independence-executions-1} says that
\[
P(\alpha | \beta) = 
P((\alpha \lceil N^1_{int}) | \beta) \times
P((\alpha \lceil N^2_{int}) | (\alpha \lceil N^1_{int}), \beta).
\]
It suffices to show both of the following:
\begin{enumerate}
\item
$P((\alpha \lceil N^1_{int}) | \beta) = P((\alpha \lceil N^1) | (\beta
  \lceil N^1))$.

For this, note that 
\[
P((\alpha \lceil N^1_{int}) | \beta) = P((\alpha \lceil N^1) | \beta),
\]
because $\beta$ already includes the firing patterns for all the
neurons in $N^1 - N^1_{int} = N^1_{ext}$.
And
\[
P((\alpha \lceil N^1) | \beta) =  
P((\alpha \lceil N^1) | (\beta \lceil N^1)),
\]
because the firing behavior of neurons in $N^1$ is independent
of the behavor of the neurons in $N - N^1$, conditioned on $\beta$.
Putting these two facts together yields the needed equality.

\item
$P((\alpha \lceil N^2_{int}) | (\alpha \lceil N^1_{int}), \beta) =
P((\alpha \lceil N^2) | (\beta \lceil N^2))$.

For this, note that
\[
P((\alpha \lceil N^2_{int}) | (\alpha \lceil N^1_{int}), \beta) =
P((\alpha \lceil N^2) | (\alpha \lceil N^1_{int}), \beta),
\]
because $\beta$ already includes the firing patterns for all the neurons in
$N^2 - N^2_{int} = N^2_{ext}$.
And
\[
P((\alpha \lceil N^2) | (\alpha \lceil N^1_{int}), \beta) =
P((\alpha \lceil N^2) | \beta),
\]
because the firing behavior of neurons in $N^2$ is independent of the
behavior of the neurons in $N^1_{int}$, conditioned on $\beta$.
Finally,
\[
P((\alpha \lceil N^2) | \beta) = 
P((\alpha \lceil N^2) | (\beta \lceil N^2)),
\]
because of locality---the neurons in $N^2$ are the only ones
that $\alpha \lceil N^2$ depends on.
Putting these three facts together yields the needed equality.
\end{enumerate}
\qed
\end{proof}

\subsection{Composition results for one-step extensions}
\label{sec: compos-one-step}

In this subsection, we describe how to break circularities in
dependencies using discrete time, as a key step toward our general
compositionality result.
In particular, we prove two lemmas showing how one-step extensions of
executions and traces of $\mathcal N$ can be expressed in terms of
one-step extensions of executions and traces of ${\mathcal N}^1$ and
${\mathcal N}^2$.

Our first lemma is about extending a finite execution, either to a particular
longer execution, or just to any execution with a particular longer
trace.

\begin{lemma}
\label{lem: independence-conditional-executions}
\begin{enumerate}
\item
Let $\alpha$ be a finite execution of $\mathcal N$ of length $>0$ that
is consistent with $\beta_{in}$.
Let $\alpha'$ be the one-step prefix of $\alpha$.
Then:
\[
P(\alpha | \alpha') =  
P((\alpha \lceil N^1_{lc}) | (\alpha' \lceil N^1)) \times  
P((\alpha \lceil N^2_{lc}) | (\alpha' \lceil N^2)).
\]
\item
Let $\beta$ be a finite trace of $\mathcal N$ of length $> 0$ that is
consistent with $\beta_{in}$.
Let $\alpha'$ be a finite execution of $\mathcal N$ such that
$trace(\alpha')$ is the one-step prefix of $\beta$.
Then:
\[
P(\beta | \alpha') =  
P((\beta \lceil N^1_{out}) | (\alpha' \lceil N^1)) \times 
P((\beta \lceil N^2_{out}) | (\alpha' \lceil N^2)).
\]
\end{enumerate}
\end{lemma}

\begin{proof}
\begin{enumerate}
\item
The non-input neurons of ${\mathcal N}$ are those in $N_{lc} =
N^1_{lc} \cup N^2_{lc}$.
The firing states of all of these neurons in the final configuration of
$\alpha$ are determined independently.
Thus, we have
\[
P(\alpha | \alpha') =  
P((\alpha \lceil N^1_{lc}) | \alpha') \times  
P((\alpha \lceil N^2_{lc}) | \alpha').
\]
Furthermore, the final firing states for the neurons in $N^1_{lc}$
depend only on the immediately previous states of the neurons in
$N^1$, and similarly for $N^2_{lc}$ and $N^2$, so this last expression
is equal to 
\[
P((\alpha \lceil N^1_{lc}) | (\alpha' \lceil N^1)) \times  
P((\alpha \lceil N^2_{lc}) | (\alpha' \lceil N^2)),
\]
as needed.

\item
The output neurons of ${\mathcal N}$ are those in $N_{out} = N^1_{out}
  \cup N^2_{out}$.
The firing states of all of these neurons in the final configuration of
$\beta$ are determined independently.
Thus, we have
\[
P(\beta | \alpha') =  
P((\beta \lceil N^1_{out}) | \alpha') \times 
P((\beta \lceil N^2_{out}) | \alpha').
\]
Furthermore, the final firing states for the neurons in $N^1_{out}$
depend only on the immediately previous states of the neurons in
$N^1$, and similarly for $N^2_{out}$ and $N^2$, so this last
expression is equal to 
\[
P((\beta \lceil N^1_{out}) | (\alpha' \lceil N^1)) \times 
P((\beta \lceil N^2_{out}) | (\alpha' \lceil N^2)),
\]
as needed.
\end{enumerate}
\qed
\end{proof}

The second lemma is about extending a finite trace, either to an execution or
to a longer trace.  This is a bit more difficult because we are
conditioning only on traces, which do not include the internal
behavior of the two sub-networks.

\begin{lemma}
\label{lem: independence-conditional-traces}
\begin{enumerate}
\item
Let $\alpha$ be a finite execution of $\mathcal N$ of length $>0$ that
is consistent with $\beta_{in}$.
Let $\beta'$ be the one-step prefix of $trace(\alpha)$.
Then:
\[
P(\alpha | \beta') =  
P((\alpha \lceil N^1_{lc}) | (\beta' \lceil N^1)) \times  
P((\alpha \lceil N^2_{lc}) | (\beta' \lceil N^2)).
\]
\item
Let $\beta$ be a finite trace of $\mathcal N$ of length $>0$ that is
consistent with $\beta_{in}$.
Let $\beta'$ be the one-step prefix of $\beta$.
Then:
\[
P(\beta | \beta') =  
P((\beta \lceil N^1_{out}) | (\beta' \lceil N^1)) \times  
P((\beta \lceil N^2_{out}) | (\beta' \lceil N^2)).
\]
\end{enumerate}
\end{lemma}

\begin{proof}
\begin{enumerate}
\item
Fix $\alpha$ and $\beta'$ as described.  
Let $\alpha'$ be the one-step prefix of $\alpha$.
By Lemma~\ref{lem: claim3}, we have: 
\[
P(\alpha | \beta') = P(\alpha | \alpha') \times P(\alpha' | \beta').
\]
Lemma~\ref{lem: independence-conditional-executions} implies that
\[
P(\alpha | \alpha') = 
P((\alpha \lceil N^1_{lc}) | (\alpha' \lceil N^1)) \times 
P((\alpha \lceil N^2_{lc}) | (\alpha' \lceil N^2)).
\]
Lemma~\ref{lem: independence-executions} implies that
\[
P(\alpha' | \beta') =
P((\alpha' \lceil N^1) | (\beta' \lceil N^1)) \times 
P((\alpha' \lceil N^2) | (\beta' \lceil N^2)).
\]
Substituting, we get that:
\[
P(\alpha | \beta') = 
P((\alpha \lceil N^1_{lc}) | (\alpha' \lceil N^1)) \times 
P((\alpha \lceil N^2_{lc}) | (\alpha' \lceil N^2)) \times
P((\alpha' \lceil N^1) | (\beta' \lceil N^1)) \times 
P((\alpha' \lceil N^2) | (\beta' \lceil N^2)).
\]
Rearranging terms and using Lemma~\ref{lem: claim4}, Part 3, we
see that the right-hand side is equal to
\[
P((\alpha \lceil N^1_{lc}) | (\beta' \lceil N^1)) \times 
P((\alpha \lceil N^2_{lc}) | (\beta' \lceil N^2)),
\]
as needed.

\item
Fix $\beta$ and $\beta'$ as described. 
Let $B$ denote the set of executions $\alpha$ of $\mathcal N$ such that 
$trace(\alpha) = \beta$, i.e., such that $\alpha \lceil N_{ext} = \beta$.
Note that what varies among the different executions in $B$ is just
the firing patterns of the neurons in $N_{int} = N^1_{int} \cup N^2_{int}$.
Then $P(\beta | \beta')$ can be expanded as 
\[
\sum_{\alpha \in B} P(\alpha | \beta').
\]
By Part 1, this is equal to
\[
\sum_{\alpha \in B} 
(P((\alpha \lceil N^1_{lc}) | (\beta' \lceil N^1)) \times 
 P((\alpha \lceil N^2_{lc}) | (\beta' \lceil N^2)).
\]

Now define $B^1$ to be the set of executions $\alpha^1$ of ${\mathcal N}^1$ 
such that $trace(\alpha^1) = \beta \lceil N^1$.
Note that all that varies among these $\alpha^1$ is the firing
patterns of the neurons in $N^1_{int}$.  Analogously, define
$B^2$ to be the set of executions $\alpha^2$ of ${\mathcal N}^2$ 
such that $trace(\alpha^2) = \beta \lceil N^2$.
All that varies among these $\alpha^2$ is the firing
patterns of the neurons in $N^2_{int}$. 

Now we project the $B$ executions onto $N^1$ and $N^2$, and we
get that the expression above is equal to:
\[
\sum_{\alpha^1 \in B^1, \alpha^2 \in B^2}
(P(\alpha^1 | (\beta' \lceil N^1)) \times P(\alpha^2 | (\beta' \lceil N^2))).
\]
This sum can be split into the product of sums:
\[
\sum_{\alpha^1 \in B^1} P(\alpha^1 | (\beta' \lceil N^1)) 
\times
\sum_{\alpha^2 \in B^2} P(\alpha^2 | (\beta' \lceil N^2)).
\]
This is, in turn, equal to
\[
P((\beta \lceil N^1_{out}) | (\beta' \lceil N^1)) \times 
P((\beta \lceil N^2_{out}) | (\beta' \lceil N^2)),
\]
as needed.
\end{enumerate}
\qed
\end{proof}

\subsection{Compositionality}
\label{sec: comp}

Finally we are ready to prove that our behavior notion $Beh$ is
compositional.
In view of Theorem~\ref{th: composition-beh2}, it suffices to show
that our auxiliary behavior notion $Beh_2$ is compositional.
And in view of Lemma~\ref{lem: compositionality-characterization}, it
suffices to show that  $Beh_2({\mathcal N})$ is uniquely determined by
$Beh_2({\mathcal N}^1)$ and $Beh_2({\mathcal N}^2)$, which we do in
Lemma~\ref{lem: determine-Beh2}.
To accomplish this, we show (in Lemma~\ref{lem:
  independence-conditional-traces-1})
how to express $Beh_2({\mathcal N})$ in
terms of $Beh_2({\mathcal N}^1)$ and $Beh_2({\mathcal N}^2)$.

Recall that the definition of $Beh_2({\mathcal N})$ specifies, for
each infinite input execution $\beta_{in}$ of $\mathcal N$, a collection of
conditional probabilities, one for each finite trace $\beta$ of
$\mathcal N$ of length $>0$ that is consistent with $\beta_{in}$.
Fix any such input execution, $\beta_{in}$, which generates a
particular probabilistic execution $P$ of ${\mathcal N}$.
Then consider an arbitrary finite trace $\beta$ of ${\mathcal N}$ of
length $t > 0$ that is consistent with $\beta_{in}$.
Let $\beta'$ be the length $t-1$ prefix of $\beta$.
We show how to express $P(\beta | \beta')$ in terms of the conditional
probabilities that arise from probability distributions $P^1$ and
$P^2$ on infinite executions of ${\mathcal N}^1$ and ${\mathcal N}^2$,
respectively.
These distributions $P^1$ and $P^2$ are defined from certain input
executions of ${\mathcal N}^1$ and ${\mathcal N}^2$, respectively.
We define these input executions and distributions as follows.
\begin{itemize}
\item
Input execution $\beta^1_{in}$ and distribution $P^1$ for ${\mathcal N}^1$:

Define the infinite input execution $\beta^1_{in}$ of ${\mathcal N}^1$
as follows.
First, note that $N^1_{in} \subseteq N_{in} \cup N^2_{out}$, that is,
every input of ${\mathcal N}^1$ is either an input of ${\mathcal N}$
or an output of ${\mathcal N}^2$.
Define the firing patterns of the neurons in $N^1_{in} \cap N_{in}$
using $\beta_{in}$, that is, define
$\beta^1_{in} \lceil (N^1_{in} \cap N_{in}) = \beta_{in} \lceil N^1_{in}$.
And for the firing patterns of the input neurons in $N^1_{in} \cap N^2_{out}$,
use $\beta'$, that is, define $\beta^1_{in} \lceil (N^1_{in} \cap
N^2_{out}) = \beta' \lceil (N^1_{in} \cap N^2_{out})$ for times
$0,\ldots,t-1$, and the default $0$ for times $\geq t$.
Define $P^1$ to be the probability distribution that is generated by
${\mathcal N}^1$ from input execution $\beta^1_{in}$.

\item
Input execution $\beta^2_{in}$ and distribution $P^2$ for ${\mathcal N}^1$:

Analogous, interchanging 1 and 2.
\end{itemize}

\begin{lemma} 
  \label{lem: compose-out}
  Define $\beta$, $\beta'$, $P^1$, and $P^2$ as above.
  Then:
\[
P(\beta | \beta') =  
P^1((\beta \lceil N^1_{out}) | (\beta' \lceil N^1)) \times  
P^2((\beta \lceil N^2_{out}) | (\beta' \lceil N^2)).
\]
\end{lemma}

\begin{proof}
Lemma~\ref{lem: independence-conditional-traces}, Part 2, tells us that:
\[
P(\beta | \beta') =  
P((\beta \lceil N^1_{out}) | (\beta' \lceil N^1)) \times  
P((\beta \lceil N^2_{out}) | (\beta' \lceil N^2)).
\]
So it suffices to show that
\[
P((\beta \lceil N^1_{out}) | (\beta' \lceil N^1)) =
P^1((\beta \lceil N^1_{out}) | (\beta' \lceil N^1)),
\]
and similarly for ${\mathcal N}^2$.

The two expressions for ${\mathcal N}^1$ look very similar; their
equivalence follows by unwinding definitions.
First, the left-hand expression is based on $P$, which is generated by
the execution of the entire network ${\mathcal N}$ for input
$\beta_{in}$.
Thus, $\beta_{in}$ defines the inputs of ${\mathcal N}^1$ that are
also inputs of ${\mathcal N}$, but not those that are outputs of
${\mathcal N}^2$---the latter emerge from $P$.
Then we consider the conditional probability $P((\beta \lceil
N^1_{out}) | (\beta' \lceil N^1))$, which means that we now assume
that the external behavior of ${\mathcal N}^1$ through time $t-1$ is
$\beta'$, and consider the (conditional) probability that the firing
pattern produced by $P$ for the outputs of ${\mathcal N}^1$ at time
$t$ coincides with what is given in $\beta$.

On the other hand, the right-hand expression is based on $P^1$, which
is generated by the execution of just the sub-network ${\mathcal N}^1$
for input $\beta^1_{in}$.
Then we consider the conditional probability $P^1((\beta \lceil
N^1_{out}) | (\beta' \lceil N^1))$, which means that we again assume
that the external behavior of ${\mathcal N}^1$ through time $t-1$ is
$\beta'$, and now consider the (conditional) probability that the
firing pattern produced by $P^1$ for the outputs of ${\mathcal N}^1$
at time $t$ coincides with what is given in $\beta$.

Note that in $P$, we may have different input sequences to ${\mathcal
  N}^1$ starting from time $t$, depending on what is produced by
network $\mathcal N$ for input $\beta_{in}$.
In $P^1$, those inputs are always $0$, as in the definition of
$\beta^1_{in}$.
This difference does not matter, because we are concerned only with
the outputs of ${\mathcal N}^1$ through time $t$, and these outputs
depend only on inputs to ${\mathcal N}^1$ through time $t-1$.

It follows that these two conditional probabilities are the same.
\qed
\end{proof}

Lemma~\ref{lem: compose-out} is a nice statement of how the
probabilities decompose, and we generalize this in Lemma~\ref{lem:
  compose-out-2}.
However, it is not quite in the right form to prove compositionality of $Beh_2$.
This is because the expressions on the right-hand-side calculate
conditional probabilities for $\beta \lceil N^1_{out}$ and $\beta
\lceil N^2_{out}$, which describe behavior of only output neurons of
the two networks,
whereas $Beh_2$ is defined in terms of probabilities for traces that
include inputs as well as outputs.
So, we need a technical modification of the lemma.

Specifically, define $\gamma^1$ to be the length-$t$ trace of
${\mathcal N}^1$
such that $\gamma^1 \lceil N^1_{out} = \beta \lceil N^1_{out}$ and 
$\gamma^1 \lceil N^1_{in}$ is a prefix of $\beta^1_{in}$.
That is, $\gamma^1$ pastes together the output from $\beta \lceil N^1_{out}$
with the input used in the definition of $P^1$.
Note that $\beta' \lceil N^1$ is the one-step prefix of $\gamma^1$.
Define $\gamma^2$ analogously.

Now we can state a lemma that expresses conditional probabilities for
${\mathcal N}$ with input $\beta_{in}$ in terms of conditional probabilities
for ${\mathcal N}^1$ with input $\beta^1_{in}$ and
${\mathcal N}^2$ with input $\beta^2_{in}$.

\begin{lemma}
  \label{lem:  independence-conditional-traces-1}
    Define $\beta$, $\beta'$, $P^1$, $P^2$, $\gamma^1$, and $\gamma^2$ as above.
  Then:
\[
P(\beta | \beta') =  
P^1(\gamma^1 | (\beta' \lceil N^1)) \times  
P^2(\gamma^2 | (\beta' \lceil N^2)).
\]
\end{lemma}

\begin{proof}
By Lemma~\ref{lem: compose-out}, we have that
\[
P(\beta | \beta') =  
P^1((\beta \lceil N^1_{out}) | (\beta' \lceil N^1)) \times  
P^2((\beta \lceil N^2_{out}) | (\beta' \lceil N^2)).
\]
So it suffices to show that the corresponding terms are the same, that
is, that:
\[
P^1((\beta \lceil N^1_{out}) | (\beta' \lceil N^1)) =
P^1(\gamma^1 | (\beta' \lceil N^1)),
\]
and similarly for ${\mathcal N}^2$.
The first case follows because the definition of $P^1$ fixes the
firing patterns for the neurons in $N^1_{in}$ through time $t$, in a
way that is consistent with $\gamma^1$, and the traces $\gamma^1$ and $\beta$
agree on the neurons in $N^1_{out}$.
Similarly for the second case.
\qed
\end{proof}

\hide{
\begin{lemma}
\label{lem:  independence-conditional-traces-1}
\[
P(\beta | \beta') =  
P^1((\beta \lceil N^1) | (\beta' \lceil N^1)) \times  
P^2((\beta \lceil N^2) | (\beta' \lceil N^2)).
\]
\end{lemma}

\begin{proof}
By Lemma~\ref{lem: compose-out}, we have that
\[
P(\beta | \beta') =  
P^1((\beta \lceil N^1_{out}) | (\beta' \lceil N^1)) \times  
P^2((\beta \lceil N^2_{out}) | (\beta' \lceil N^2)).
\]
So it suffices to show that the corresponding terms are the same, that
is:
\[
P^1((\beta \lceil N^1_{out}) | (\beta' \lceil N^1)) =
P^1((\beta \lceil N^1) | (\beta' \lceil N^1)),
\]
and similarly for ${\mathcal N}^2$.
This follows because the definition of $P^1$ fixes the firing patterns
for the neurons in $N^1_{in}$ through time $t$ 
(based on $\beta_{in}$ through time $t$ and $\beta' \lceil N^1_{in}$
followed by $0$).
\qed
\end{proof}
}

Now we can conclude compositionality:

\begin{lemma}
\label{lem: determine-Beh2}
For all compatible pairs of networks ${\mathcal N}^1$ and ${\mathcal N}^2$, 
$Beh_{2}({\mathcal N})$ is determined by $Beh_{2}({\mathcal N}^1)$ and
$Beh_{2}({\mathcal N}^2)$.
\end{lemma}

\begin{proof}
Follows directly from Lemma~\ref{lem:
  independence-conditional-traces-1}.
\qed
\end{proof}

\begin{theorem}
\label{th: beh2-compositional}
$Beh_{2}$ is compositional.
\end{theorem}

\begin{proof}
By Lemmas~\ref{lem: determine-Beh2} and ~\ref{lem:
  compositionality-characterization}.
\qed
\end{proof}

\begin{theorem}
\label{th: beh-compositional}
$Beh$ is compositional.
\end{theorem}

\begin{proof}
By Theorems~\ref{th: beh2-compositional} and~\ref{th:
  composition-beh2}.
\qed
\end{proof}

We end this section with a generalization of Lemma~\ref{lem:
  compose-out}
that applies to all four combinations of executions and traces.
The proof is similar to that for Lemma~\ref{lem: compose-out},
based on earlier Lemmas~\ref{lem: independence-conditional-executions}
and~\ref{lem: independence-conditional-traces}.
We will use this in Section~\ref{example: cyclic-2}.

\begin{lemma}
  \label{lem: compose-out-2}
  Let $\alpha$ be a finite execution of $\mathcal N$ of length $>0$
  that is consistent with $\beta_{in}$.
  Let $\alpha'$ be its one-step prefix.
  Let $\beta = trace(\alpha)$ and $\beta' = trace(\alpha')$.
Let $P_1$ and $P_2$ be as defined earlier in this section.
Then
\begin{enumerate}
\item
$P(\alpha | \alpha') =  
P^1((\alpha \lceil N^1_{lc}) | (\alpha' \lceil N^1)) \times  
P^2((\alpha \lceil N^2_{lc}) | (\alpha' \lceil N^2))$.
\item
$P(\beta | \alpha') =  
P^1((\beta \lceil N^1_{out}) | (\alpha' \lceil N^1)) \times  
P^2((\beta \lceil N^2_{out}) | (\alpha' \lceil N^2))$.
\item
$P(\alpha | \beta') =  
P^1((\alpha \lceil N^1_{lc}) | (\beta' \lceil N^1)) \times  
P^2((\alpha \lceil N^2_{lc}) | (\beta' \lceil N^2))$.
\item
$P(\beta | \beta') =  
P^1((\beta \lceil N^1_{out}) | (\beta' \lceil N^1)) \times  
P^2((\beta \lceil N^2_{out}) | (\beta' \lceil N^2))$.
\end{enumerate}
\end{lemma}

\subsection{Examples}


\subsubsection{Toy example for cyclic composition}
\label{example: cyclic-2}

We consider the toy cyclic composition example from
Section~\ref{example: cyclic-1}.
We analyze just one case in detail, namely, where $x_1$ fires
at time $0$ and $x_2$ does not.
We prove that, with probability at least $(1-\delta)^{7}$, both $x_1$
and $x_2$ fire at time $4$.

The input firing sequence $\beta_{in}$ is trivial here, since the
composed network $\mathcal N$ has no input neurons.
For this example, we assume that, in the initial configuration, $x_1$
fires and the other three neurons do not fire.
With these restrictions, we have a single probability distribution $P$
for infinite executions of $\mathcal N$.
We argue compositionally, in terms of executions.

So let $E$ be the set of executions of length $4$ in which both $x_1$
and $x_2$ fire at time $4$.
We will show that $P(E) \geq (1-\delta)^{7}$.
For this, we define several other sets of executions.
Each set is included in the previous one.
\begin{itemize}
\item
$E_0$, the set of executions of length $0$ consisting of just the
  initial configuration, in which $x_1$ is firing and the other neurons are
  not firing.
\item
$E_1$, the set of executions of length $1$ whose one-step prefix is in
  $E_0$ and in which, in the last configuration, $a_1$ is firing.
\item
$E_2$, the set of executions of length $2$ whose one-step prefix is in
  $E_1$ and in which, in the last configuration, $x_2$ is firing.
\item
$E_3$, the set of executions of length $3$ whose one-step prefix is in
  $E_2$ and in which, in the last configuration, $x_2$ and $a_2$ are
  both firing.
\item
$E_4$, the set of executions of length $4$ whose one-step prefix is in
  $E_3$ and in which, in the last configuration, $x_1$, $x_2$ and $a_2$ are
  all firing.
\end{itemize}
%
%
Then we can see that 
\[
P(E) \geq P(E_4) = P(E_4 | E_3) P(E_3 | E_2) P(E_2 | E_1) P(E_1 | E_0)
P(E_0)
= P(E_4 | E_3) P(E_3 | E_2) P(E_2 | E_1) P(E_1 | E_0).
\]

We need lower bounds for the four conditional probabilities.
For example, consider $P(E_4 | E_3)$.
Let $\alpha'$ be any execution in $E_3$; we will argue that 
$P(E_4 | \alpha') \geq (1 - \delta)^3$, and use Total Probability to conclude
that $P(E_4 | E_3) \geq (1 - \delta)^3$.
We have:
\[
P(E_4 | \alpha') = \sum_{\alpha} P(\alpha | \alpha'),
\]
where $\alpha$ ranges over the length-$4$ executions in $E_4$ that extend
$\alpha'$.
By Lemma~\ref{lem: compose-out-2}, we may break this down in terms of
the two sub-networks and write:
\[
P(\alpha | \alpha') = 
P^1((\alpha \lceil N^1_{lc}) | (\alpha' \lceil N^1)) \times  
P^2((\alpha \lceil N^2_{lc}) | (\alpha' \lceil N^2)),
\]
where $P^1$ and $P^2$ are defined from $\beta' = trace(\alpha')$ as in
Section~\ref{sec: comp}.

We can rewrite $\sum_{\alpha} P(\alpha | \alpha')$ as
\[
\sum_{\alpha^1} \sum_{\alpha^2} 
P^1((\alpha^1 \lceil N^1_{lc}) | (\alpha' \lceil N^1)) \times  
P^2((\alpha^2 \lceil N^2_{lc}) | (\alpha' \lceil N^2)),
\]
where $\alpha^1$ ranges over all one-step extensions of $\alpha'
\lceil N^1$ such that $x_2$ fires in the final configuration,
and $\alpha^2$ ranges over all one-step extensions of 
$\alpha' \lceil N^2$ in which $x_1$ and $a_2$ both fire in the final
configuration.
%
%
This summation is equal to
\[
\sum_{\alpha^1} P^1((\alpha^1 \lceil N^1_{lc}) | (\alpha' \lceil N^1)) \times  
\sum_{\alpha^2} P^2((\alpha^2 \lceil N^2_{lc}) | (\alpha' \lceil N^2)).
\]
The first term is $\geq (1 - \delta)$ because we care only that $x_2$
fires in the final configuration, and we have assumed that it fires in
the previous configuration.
The second term is $\geq (1 - \delta)^2$, because we care that both
$x_1$ and $a_2$ fire in the final configuration, and we have assumed
that $a_2$ and $x_2$ fire in the previous configuration.
So we have:
\[
P(E_4 | \alpha') =
\sum_{\alpha^1} P^1((\alpha^1 \lceil N^1_{lc}) | (\alpha' \lceil N^1)) \times  
\sum_{\alpha^2} P^2((\alpha^2 \lceil N^2_{lc}) | (\alpha' \lceil N^2))
\geq (1-\delta)(1-\delta)^2 = (1-\delta)^3.
\]

Thus, we have shown that
$P(E_4 | E_3) \geq (1 - \delta)^3$,
Similar arguments can be used to show that 
$P(E_3 | E_2) \geq (1 - \delta)^2$,
$P(E_2 | E_1) \geq (1 - \delta)$, and
$P(E_1 | E_0) \geq (1 - \delta)$.
Combining all the terms we get that $P(E_4) \geq (1 - \delta)^7$, as
needed.

\section{Hiding for Spiking Neural Networks}
\label{sec: hiding}

Now we define our second operator for SNNs, the \emph{hiding
  operator}.
This operator is designed to ``hide'' some previously
externally-visible behavior so it becomes invisible outside the
network.
Formally, the hiding operator simply reclassifies some output neurons
as internal.
The hiding operator can be used in conjunction with a composition
operator; for example, we often want to compose two networks and then
hide the neurons that were used to communicate between them.

\subsection{Hiding definition}

Given a network ${\mathcal N}$ and a subset $V$ of the output neurons
$N_{out}$ of ${\mathcal N}$, we define a new network ${\mathcal N}' =
hide({\mathcal N}, V)$ to be exactly the same as ${\mathcal N}$ except
that all the outputs in $V$ are now reclassified as internal neurons.  
That is, all parts of the definition of ${\mathcal N}'$ and
${\mathcal N}$ are identical except that 
$N'_{out} = N_{out} - V$ and
$N'_{int} = N_{int} \cup V$.
The effect of the hiding operator is to make the hidden
neurons ineligible for combining with other neurons in further
composition operations.

We give a result in the style of Lemma~\ref{lem: determine-Beh2},
here saying that the external behavior of 
$hide({\mathcal N}, V)$ is determined by the external behavior of
$\mathcal N$ and $V$.

\begin{theorem}
For every network ${\mathcal N}$ and subset $V \subseteq N_{out}$,
$Beh(hide({\mathcal N},V))$ is determined by $Beh(N)$ and $V$.
\end{theorem}

\begin{proof}
Let ${\mathcal N}'=hide({\mathcal N}, V)$.
Fix any infinite input execution $\beta_{in}$ for ${\mathcal N}'$, and
let $P'$ denote the probabilistic execution of ${\mathcal N}'$
generated from $\beta_{in}$.
Consider any finite trace $\beta$ of ${\mathcal N}'$ that is
consistent with $\beta_{in}$.
We must express $P'(\beta)$ in terms of the probability distribution
of traces generated by $\mathcal N$ on some input execution.

To do this, note that the executions of $\mathcal N$ are identical to
those of ${\mathcal N}'$---only the classification of neurons in $V$
is different.  
In particular, the input execution $\beta_{in}$ is also an input execution of
$\mathcal N$.
Let $P$ denote the probabilistic execution generated of $\mathcal N$
generated from $\beta_{in}$.
Then $P'$, the probabilistic execution of ${\mathcal N}'$, is identical
to $P$, the probabilistic execution of $\mathcal N$.
So we can write $P'(\beta) = P(\beta)$.

This is not quite what we need, because $\beta$ is not actually a
trace of ${\mathcal N}$---it excludes firing patterns for neurons in $V$.
But we can define $B$ to be the set of traces $\gamma$ of ${\mathcal N}$
such that $\gamma \lceil (N'_{ext}) = \beta$, that is, $B$ is the
set of traces of ${\mathcal N}$ that project to yield $\beta$ but
allow any firing behavior for the neurons in $V$.
Then we have 
\[
P'(\beta) = \sum_{\gamma \in B} P(\gamma).
\]
This is enough to show the needed dependency.
\qed
\end{proof}

\subsection{Examples}


\subsubsection{Boolean circuits}
\label{example: Boolean-circuits-4}

Let $\mathcal N$ be the $5$-gate Nand circuit from
Section~\ref{example: Boolean-circuits-2}.
Let $V$ be the singleton set consisting of just the And neuron within
the circuit.
We consider the network ${\mathcal N}' = hide({\mathcal N},V)$, which
is the same as the Nand circuit except that the And neuron is now
regarded as internal.
Thus, ${\mathcal N}'$ has two internal neurons:  the And neuron, and
the internal neuron $a$ of $\mathcal N$.
Fix $\beta_{in}$ to be any infinite input execution (for both
${\mathcal N}$ and ${\mathcal N}'$) with stable inputs, and let $P$
and $P'$ be the probabilistic executions of $\mathcal N$ and
${\mathcal N}'$, respectively, generated from $\beta_{in}$.

In $P'$, we should expect to have stable correct Nand outputs for a long
time starting from time $3$.
Here we consider just finite traces $\beta$ of length exactly $3$, and
focus on the output at exactly time $3$.
Thus, we consider the probabilities $P'(\beta)$ for finite traces
$\beta$ of length exactly $3$, and we would like to show that the
probability of a correct Nand output at time $3$ is at least $(1 - \delta)^3$.
We use the connection between $P$ and $P'$ to help us show this.

Namely, we assume that, in $P$, the probability of both a correct And
output at time $1$ and a correct Nand output at time $3$ is at
least $(1 - \delta)^3$.
This could be proved for the Nand circuit separately, but we simply
assume it here.

Now define event $B$ to be the set of traces $\beta$ of ${\mathcal
  N}'$ of length $3$ such that $\beta$ gives a correct Nand output at
time $3$.  
Our assumption about $P$ implies that $P(B) \geq (1 - \delta)^3$.
We argue that $P'(B) \geq (1 - \delta)^3$, which implies our desired
result.

We have that $P'(B) = \sum_{\beta \in B}P'(\beta)$.
We know that $P'(\beta) = P(\beta)$ for each trace $\beta$ of
${\mathcal N}'$.
Therefore, we have that $P'(B) = \sum_{\beta \in B}P(\beta) = P(B)$.
Since we have that $P(B) \geq (1 - \delta)^3$, it follows that
$P'(B) \geq (1 - \delta)^3$, as needed.

\section{Problems for Spiking Neural Networks}
\label{sec: problem}

In this section, we define a formal notion of a \emph{problem} to be
solved by a stochastic Spiking Neural Network.  Problems are stated in
terms of the input/output behavior that should be exhibited by a
network.
Namely, for every input, a problem specifies a set of \emph{possibilities},
each of which is a probability distribution on outputs.
We define what it means for an SNN to \emph{solve} a problem.  We
prove that this notion of ``solves'' respects our composition and hiding
operators.
 

\subsection{Problems and solving problems}

We define a \emph{problem} ${\mathcal R}$ for a pair
$(N_{in},N_{out})$ of disjoint sets of neurons to be a mapping that
assigns, to each infinite sequence $\beta_{in}$ of firing patterns for
$N_{in}$, a nonempty set ${\mathcal R}(\beta_{in})$ of \emph{possibilities}.
Each possibility $R \in {\mathcal R}(\beta_{in})$ is a mapping that
specifies, for every finite sequence $\beta$ of firing patterns for
$N_{in} \cup N_{out}$ that is consistent with $\beta_{in}$, a probability $R(\beta)$.
Thus, the problem $\mathcal R$ assigns to each input a set of ``possible''
probability distributions on outputs.

The probabilities assigned by a particular possibility $R$ must satisfy
certain constraints, designed to guarantee that they generate an actual
probability distribution on the set of infinite sequences of firing 
patterns for $N_{in} \cup N_{out}$.
Namely, we require that $R$ assign probability $1$ to some particular
$\beta$ of length $0$, and that the probabilities assigned to the
one-step extensions of any $\beta$ must add up to the probability of
$\beta$.

Now suppose that ${\mathcal N}$ is a network with input and output
neurons $N_{in}$ and $N_{out}$, and $\mathcal R$ is a problem for
$(N_{in}, N_{out})$.
Then we say that that ${\mathcal N}$ \emph{solves} ${\mathcal R}$
provided that, for any infinite input execution $\beta_{in}$ for $\mathcal N$,
there is some possibility $R \in {\mathcal R}(\beta_{in})$ for which
the following holds:
Let $P$ denote the probabilistic execution of ${\mathcal N}$ for $\beta_{in}$.
Then for every finite trace $\beta$ of $\mathcal N$, $P(\beta) = R(\beta)$.
In other words, $R$ is exactly the trace distribution derived from the
probabilistic execution of $\mathcal N$ for input $\beta_{in}$.

\subsection{Composition of problems}

We would like a theorem of the following form:
If ${\mathcal N}^1$ solves problem ${\mathcal R}^1$ and ${\mathcal N}^2$ 
solves problem ${\mathcal R}^2$, then the composition of networks
${\mathcal N} = {\mathcal N}^1 \times {\mathcal N}^2$ solves
the composition of problems ${\mathcal R} = {\mathcal R}^1 \times
{\mathcal R}^2$.
For this, we must first define the composition of two problems, 
${\mathcal R} = {\mathcal R}^1 \times {\mathcal R}^2$.

So let ${\mathcal R}^1$ be a problem for the pair $(N^1_{in}, N^1_{out})$
and ${\mathcal R}^2$ a problem for the pair $(N^2_{in}, N^2_{out})$.
Assume that ${\mathcal R}^1$ and ${\mathcal R}^2$ are
\emph{compatible}, in the sense that $N^1_{out} \cap N^2_{out} = \emptyset$.
Then the composition $\mathcal R$ is defined to be a problem for the
pair $(N_{in}, N_{out})$, where
$N_{out} = N^1_{out} \cup N^2_{out}$ and
$N_{in} = N^1_{in} \cup N^2_{in} - N_{out}$. 
%
The composed problem $\mathcal R$ should be defined as a mapping that
assigns, to each infinite sequence $\beta_{in}$ of firing patterns for $N_{in}$, a
nonempty set ${\mathcal R}(\beta_{in})$ of \emph{possibilities}.
Each possibility $R \in {\mathcal R}(\beta_{in})$ should be a mapping that
specifies, for every finite sequence $\beta$ of firing patterns for
$N_{in} \cup N_{out}$ that is consistent with $\beta_{in}$, a
probability $R(\beta)$.

We define the ${\mathcal R}$ mapping by considering each $\beta_{in}$
separately; so fix any $\beta_{in}$.
We describe how to define the set ${\mathcal R}(\beta_{in})$ of
possibilities for $\beta_{in}$.

To define ${\mathcal R}(\beta_{in})$, we start by selecting (in an
arbitrary way) a single possibility
$R^1(\beta^1_{in}) \in{\mathcal R}^1(\beta^1_{in})$ for each firing
pattern $\beta^1_{in}$ for $N^1_{in}$, and likewise a single
possibility
$R^2(\beta^2_{in}) \in {\mathcal R}^2(\beta^2_{in})$ for each firing
pattern $\beta^2_{in}$ for $N^2_{in}$.\footnote{Unwinding the
  definitions a bit, possibility $R^1(\beta^1_{in})$ is a
  mapping from sequences of firing patterns that are consistent with
$\beta^1_{in}$ to probabilities, and analogously for $R^2(\beta^2_{in})$.}
We use this entire collection of choices for $R^1(\beta^1_{in})$ and
$R^2(\beta^2_{in})$, for all values of $\beta^1_{in}$ and $\beta^2_{in}$, to
construct a single, particular possibility $R$ for $\beta_{in}$.
Then we define ${\mathcal R}(\beta_{in})$ to be the set of all
possibilities for $\beta_{in}$ that can be constructed in this way,
based on all choices for the possibilities $R^1(\beta^1_{in})$ and
$R^2(\beta^2_{in})$.

So fix the possibilities $R^1(\beta^1_{in}) \in{\mathcal R}^1(\beta^1_{in})$ and
$R^2(\beta^2_{in}) \in {\mathcal R}^2(\beta^2_{in})$ arbitrarily, as
just described.
Constructing the possibility $R$ for $\beta_{in}$ requires us to
define $R(\beta)$ for every finite sequence $\beta$ of firing patterns
of $N_{in} \cup N_{out}$ that is consistent with $\beta_{in}$.
We do this recursively.
For the base, consider $\beta$ of length $0$, where $\beta$ is
consistent with $\beta_{in}$.
Let $\beta^1_{in}$ be the infinite sequence of all-$0$ firing patterns for
$N^1_{in}$, and $\beta^2_{in}$ be the infinite sequence of all-$0$
firing patterns for $N^2_{in}$.
Then we define $R(\beta) = 1$ if
\[
R^1(\beta^1_{in})(\beta \lceil N^1_{out}) = 1 \mbox{ and }
R^2(\beta^2_{in})(\beta \lceil N^2_{out}) = 1,
\]
and $0$ otherwise.
That is, we assign probability $1$ to the length-$0$ sequence $\beta$
that is consistent with $\beta_{in}$, and in which the output firing
states are the same as those to which $R^1(\beta^1_{in})$ and
$R^2(\beta^2_{in})$ assign probability $1$.

For the recursive step, consider $\beta$ of length $\geq 1$, where
$\beta$ is consistent with $\beta_{in}$, and let $\beta'$ be the
one-step prefix of $\beta$.
We define $R(\beta)$ in terms of $R(\beta')$.
Namely, let $\beta^1_{in}$ be the infinite sequence of firing patterns for
$N^1_{in}$ that are constructed from the following:
(a) for neurons in $N^1_{in} \cap N_{in}$, use  $\beta_{in} \lceil N^1_{in}$, and
(b) for neurons in $N^1_{in} \cap N^2_{out}$, use  $\beta' \lceil (N^1_{in} \cap
N^2_{out})$ for times $0,\ldots,t-1$, and the default $0$ for times $\geq t$, 
Define $\beta^2_{in}$ analogously.
%
Then define $R(\beta) = R(\beta') \times T^1 \times T^2$,
where $T^1$ is the conditional probability
$R^1(\beta^1_{in})((\beta \lceil N^1_{out}) | (\beta' \lceil N^1))$
and $T^2$ is the conditional probability
$R^2(\beta^2_{in})((\beta \lceil N^2_{out}) | (\beta' \lceil
N^2))$.\footnote{Again unwinding the definitions,
  $R^1(\beta^1_{in})$ is the possibility chosen for input
  $\beta^1_{in}$.
  The conditional probability
  $R^1(\beta^1_{in})((\beta \lceil N^1_{out}) | (\beta' \lceil N^1))$
  describes the probability that ${\mathcal N}^1$ extends $\beta'
  \lceil N^1$
  to yield the outputs specified by $\beta$.
  Analogously for $T^2$.}

\begin{theorem}
\label{th: composition-problems}
If ${\mathcal N}^1$ solves problem ${\mathcal R}^1$ and ${\mathcal N}^2$ 
solves problem ${\mathcal R}^2$, then the composition of networks
${\mathcal N} = {\mathcal N}^1 \times {\mathcal N}^2$ solves
the composition of problems ${\mathcal R} = {\mathcal R}^1 \times
{\mathcal R}^2$.
\end{theorem}

\begin{proof}
Since ${\mathcal N}^1$ solves ${\mathcal R}^1$, we know that, for
every infinite input execution $\beta^1_{in}$ for ${\mathcal N}^1$,
there is a possibility in ${\mathcal R}^1(\beta^1_{in})$ that is
identical to the trace distribution derived from the probabilistic
execution of ${\mathcal N}^1$ for $\beta^1_{in}$.
Denote this possibility by $R^1(\beta^1_{in})$.
Likewise, since ${\mathcal N}^2$ solves ${\mathcal R}^2$, we know
that, for every infinite input execution $\beta^2_{in}$ for ${\mathcal N}^2$,
there is a possibility in ${\mathcal R}^2(\beta^2_{in})$ that is
identical to the trace distribution derived from the probabilistic
execution of ${\mathcal N}^2$ for input $\beta^2_{in}$.
Denote this possibility by $R^2(\beta^2_{in})$.
To show that $\mathcal N$ solves $\mathcal R$, we must show that, for
every infinite input execution $\beta_{in}$ for $\mathcal N$, 
there is some possibility $R \in {\mathcal R}(\beta_{in})$ such that 
$R$ is identical to the trace distribution derived from the
probabilistic execution of $\mathcal N$ for input $\beta_{in}$.

So fix an input execution $\beta_{in}$ for  $\mathcal N$, and
define $P$ to be the trace distribution generated by $\mathcal N$ for
input $\beta_{in}$.
Also define distribution $R$ for $\beta_{in}$ using the recursive
approach in the definition of composition of problems, but now based
on the particular selections $R^1$ and $R^2$ just defined.
We claim that $P = R$.
To show this, we must show that, for any finite trace $\beta$ of
$\mathcal N$ that is consistent with $\beta_{in}$, $P(\beta) = R(\beta)$.
We do this by induction on the length of $\beta$.

For the base, consider $\beta$ of length $0$.
The definition of $P(\beta)$ yields $1$ if $\beta$ is the initial output
configuration of ${\mathcal N}$ and $0$ otherwise.
The initial output configuration is the unique configuration $C$ for
which
$C \lceil N^1_{out} = F^1_0 \lceil N^1_{out}$ and 
$C \lceil N^2_{out} = F^2_0 \lceil N^2_{out}$
(here using the general notation for initial firing patterns).
On the other hand, the definition of $R(\beta)$ yields $1$ if $\beta$
is the unique output configuration of $\mathcal N$ for which
$R^1(\beta^1_{in})(\beta \lceil N^1_{out}) = 1 \mbox{ and }
R^2(\beta^2_{in})(\beta \lceil N^2_{out}) = 1$,
where $\beta^1_{in}$ and $\beta^2_{in}$ are infinite sequences of
all-$0$ firing patterns, and $0$ for other output configurations.
By definition of $R^1$ and $R^2$, this is, again, just the initial output
configuration of $\mathcal N$.
This implies that $P(\beta) = R(\beta)$.

For the inductive step, consider $\beta$ of length $\geq 1$, and
let $\beta'$ be the one-step prefix of $\beta$.
By the inductive hypothesis, we may assume that $P(\beta') = R(\beta')$.
We must show that $P(\beta) = R(\beta)$.

Fix $\beta^1_{in}$ and $\beta^2_{in}$ as in the recursive definition
of $R(\beta)$.
Then by the definition of $R(\beta)$, we have
\[
R(\beta) = R(\beta') \times 
R^1(\beta^1_{in})((\beta \lceil N^1_{out}) | (\beta' \lceil N^1)) \times
R^2(\beta^2_{in})((\beta \lceil N^2_{out}) | (\beta' \lceil N^2)).
\]
Also, for the same $\beta^1_{in}$ and $\beta^2_{in}$, fix $P^1$
and $P^2$ to be the probabilistic traces for ${\mathcal N}^1$ and
${\mathcal  N}^2$, respectively.
Then by Lemma~\ref{lem: compose-out} and Lemma~\ref{lem: subsets}, we
have
\[
P(\beta) = P(\beta') \times
P^1((\beta \lceil N^1_{out}) | (\beta' \lceil N^1)) \times  
P^2((\beta \lceil N^2_{out}) | (\beta' \lceil N^2)).
\]

The assumption that ${\mathcal N}^1$ solves ${\mathcal R}^1$ with
the particular possibility $R^1(\beta^1_{in})$ implies that the two
conditional distributions $P^1$ and $R^1(\beta^1_{in})$ are identical,
so
\[
P^1((\beta \lceil N^1_{out}) | (\beta' \lceil N^1)) = 
R^1(\beta^1_{in})((\beta \lceil N^1_{out}) | (\beta' \lceil N^1)).
\]
Similarly, $P^2$ and $R^2(\beta^2_{in})$ are identical, so
\[
P^2((\beta \lceil N^2_{out}) | (\beta' \lceil N^2)) = 
R^2(\beta^2_{in})((\beta \lceil N^2_{out}) | (\beta' \lceil N^2)).
\]
Since all three pairs of corresponding terms in the two equations are 
equal, we conclude that their products are equal, that is, $P(\beta) =
R(\beta)$, as needed.
\qed
\end{proof}

\subsection{Hiding of problems}

Next, we define a hiding operator on problems, analogous to the hiding
operator on networks.
Namely, given a problem ${\mathcal R}$ for $(N_{in}, N_{out})$, and a
subset $V$ of the output neurons $N_{out}$ of ${\mathcal R}$, we
define a new ``hidden'' problem ${\mathcal R}' = hide({\mathcal R},
V)$ for $(N'_{in}, N'_{out})$, where 
$N'_{out} = N_{out} - V$ and
$N'_{in} = N_{in}$.
%
The hidden problem ${\mathcal R}'$ should be defined as a mapping that
assigns, to each infinite sequence $\beta_{in}$ of firing patterns for
${\mathcal  N}'$, a nonempty set ${\mathcal R}'(\beta_{in})$ of
\emph{possibilities}.
Each possibility $R' \in {\mathcal R}'(\beta_{in})$  should be a mapping
that specifies, for every finite sequence $\beta$ of firing patterns
for $N'_{in} \cup N'_{out}$ that is consistent with $\beta_{in}$, a
probability $R'(\beta)$.

We define this mapping by considering each $\beta_{in}$ separately; so
fix any $\beta_{in}$.
To define the set ${\mathcal R}'(\beta_{in})$, we start by selecting
(in an arbitrary way) a single possibility $R \in {\mathcal R}(\beta_{in})$
We use $R$ to define the possibility $R'$ for ${\mathcal N}'$
and input $\beta_{in}$.
Since there may be many ways to define $R$, ${\mathcal R}'$ may wind
up containing many different possibilities.

Constructing the possibility $R'$ requires us to define $R'(\beta)$ for
every finite sequence $\beta$ of firing patterns of $N'_{in} \cup N'_{out}$ that is
consistent with $\beta_{in}$.
This construction is much simpler than that for composition:  
Let $B$ denote the set of finite sequences $\gamma$ of firing
patterns for $N_{ext}$ such that $\gamma \lceil (N'_{in} \cup N'_{out}) = \beta$.
Then define
\[
R'(\beta) = \sum_{\gamma \in B} R(\gamma).
\]

\begin{theorem}
If network $\mathcal N$ solves problem $\mathcal R$, and $V \in N_{out}$,
then network ${\mathcal N}' = hide({\mathcal N}, V)$ solves problem
${\mathcal R}' = hide({\mathcal R}, V)$.
\end{theorem}

\begin{proof}
  Since $\mathcal N$ solves $\mathcal R$, we know that, for every
  infinite input execution $\beta_{in}$ for $\mathcal N$, there is a
  possibility in ${\mathcal R}(\beta_{in})$ that is identical to the
  trace distribution derived from execution of  $\mathcal N$ for input
  $\beta_{in}$.
  Denote this possibility by $R(\beta_{in})$.
  To show that ${\mathcal N}'$ solves ${\mathcal R}'$, we must show
  that, for every input execution $\beta_{in}$ for ${\mathcal N}'$,
  there is some possibility in ${\mathcal R}'(\beta_{in})$ that is
  identical to the trace distribution derived from the probabilistic
  execution of ${\mathcal N}'$ for input $\beta_{in}$.

  So fix an input execution $\beta_{in}$ for ${\mathcal N}'$.
  Define $P'$ to be the trace distribution generated by ${\mathcal N}'$
  for input $\beta_{in}$.
  Also define distribution $R'$ for $\beta_{in}$ as in the definition
  of hiding of problems, now based on the particular
  selection $R(\beta_{in})$ just defined.
  We claim that $P' = R'$.
  This means that for any finite trace $\beta$ of ${\mathcal N}'$ that
  is consistent with $\beta_{in}$, $P'(\beta) = R'(\beta)$.
  
  To see this, let $B$ denote the set of finite sequences $\gamma$ of
  firing patterns for $N_{in} \cup N_{out}$ such that $\gamma \lceil
  (N'_{in} \cup N'_{out}) = \beta$.
  Then $P'(\beta) = \sum_{\gamma \in B} P(\gamma)$ and
  $R'(\beta) = \sum_{\gamma \in B} R(\beta_{in})(\gamma)$.
  Since $\mathcal N$ solves $\mathcal R$ with the particular
  possibility $R(\beta_{in})$, it follows that for each such $\gamma$,
  $P(\gamma) = R(\beta_{in})(\gamma)$.
  Consequently, the two summations are equal, as needed.
  \qed
\end{proof}

\subsection{Examples}
\label{example: WTA-4}

In this section, we define three problems satisfying our formal
definition of problems.
They are the Winner-Take-All (WTA) problem, the Filter problem, and an
Attention problem that can be solved by combining solutions to the WTA
and Filter problems.

\subsubsection{The Winner-Take-All problem}

We define the Winner-Take-All problem formally using notation that
corresponds to the statement of Theorem~\ref{th: WTA}:  we write it as
$WTA(n, \delta, t_c, t_s)$, using four parameters from the theorem
statement.
The problem definition allows considerable freedom, in the
choice of which output ends up firing, in the time when the
stable interval begins, and in what happens outside the stable interval.

The set $N_{in}$ is $\{ x_1,\ldots,x_n \}$, and $N_{out}$ is
$\{y_1,\ldots,y_n \}$.
For each infinite sequence $\beta_{in}$ of firing patterns for
$N_{in}$, the $WTA$ problem specifies a set of probability
distributions on sequences of firing patterns for $N_{in} \cup
N_{out}$ that are consistent with $\beta_{in}$.

So consider any particular $\beta_{in}$.  If the firing pattern for
$N_{in}$ in $\beta_{in}$ is not stable or does not have at least one
firing neuron, then we allow all distributions that are consistent
with $\beta_{in}$.
Now consider the case where $\beta_{in}$ is stable with at least one
firing neuron.
Then the possibilities for $\beta_{in}$ are exactly the
distributions that satisfy the following condition:  
With probability $\geq 1 - \delta$, there is some $t \leq t_c$ such
that the $y$ outputs stabilize by time $t$ to one steadily-firing
output $y_i$, and this firing pattern persists through time $t + t_s - 1$.
Notice that these distributions may differ in many ways, for example,
they may give equal probabilities to each output choice, or may
favor some over others.  They may exhibit different times, or probability
distributions of times, for when the stable interval begins.
They may exhibit different types of behavior before and after the
stable interval.

We argue that our $WTA$ network from Section~\ref{example: WTA-1}
solves the formal problem $WTA(n, \delta, t_c, t_s)$.
Specifically, we consider our network with the weighting
factor $\gamma$ satisfying the inequality
$\gamma \geq c_1 \log(\frac{n t_s}{\delta})$, and with $t_c \approx c_2
\log n \log(\frac{1}{\delta})$.
And we allow initial firing patterns for the internal and output
neurons to be arbitrary; so technically, we are talking about a class
of networks, not a single network.
Then Theorem~\ref{th: WTA} implies that each of these networks solves
the $WTA(n, \delta, t_c, t_s)$ problem.

\subsubsection{The Filter problem}

We define the Filter problem as $Filter(n, \delta)$.  The
set $N_{in}$ is $\{w_i, y_i | 1 \leq i \leq n \}$ and the set
$N_{out}$ is $\{z_i | 1 \leq i \leq n \}$.
The Filter problem is intended to say that, for every $i$, $1 \leq i
\leq n$, the output neuron $z_i$ should fire at any time $t \geq 1$
exactly if both the corresponding inputs $w_i$ and $y_i$ fired at time
$t-1$.  Thus, it acts like $n$ And networks.

Formally, for each infinite sequence $\beta_{in}$ of firing patterns
for $N_{in}$, the $Filter(n,\delta)$ problem specifies a set of probability
distributions on sequences of firing patterns for $N_{in} \cup
N_{out}$ that are consistent with $\beta_{in}$.

So consider any particular $\beta_{in}$.
Then the possibilities for $\beta_{in}$ are exactly the
distributions that satisfy the following condition, here expressed in
terms of conditional probabilities (which could be translated into
absolute probabilities):
Let $\beta$ be any finite sequence over $N_{in} \cup N_{out}$ of
length $t \geq 1$ that is consistent with $\beta_{in}$, and let $C_t$
be the final configuration of $\beta$.
Let $\beta'$ be the the one-step prefix of $\beta$, and $C_{t-1}$ be
the final configuration of $\beta'$.
Suppose that, for every $i$, $1 \leq i \leq n$, $C_t(z_i) =
C_{t-1}(w_i) \wedge C_{t-1}(y_i)$.
That is, $\beta$ extends $\beta'$ with correct outputs at the final
time $t$.
Then $P(\beta | \beta') \geq 1-\delta$.
%
%
The differences among these distributions may involve different
conditional probabilities (for example, different for different
outputs), as long as they satisfy the given inequality.

Our simple $Filter$ network of Section~\ref{example: WTA-2} solves the
formal $Filter$ problem, with $\delta = 1 - (1 - \delta')^n$,
where $\delta'$ is the failure probability for a single And gate at a
single time, according to notation used in Section~\ref{example: WTA-2}.

\subsubsection{The Attention problem}

We define the Attention problem formally as 
\[
Attention(n, \delta, t_c, t_s) = WTA(n, \delta', t_c, t_s) \times Filter(n, \delta'').
\]
Here $\delta$, $\delta'$, and $\delta''$ are related so that
$(1 - \delta) = (1 - \delta') (1 - \delta'')^{t_s}$.
The set $N_{in}$ is $\{ x_i, w_i | 1 \leq i \leq n \}$, and
$N_{out}$ is $\{y_i, z_i | 1 \leq i \leq n \}$.

By the definition of composition of problems, the guarantees of
$Attention(n, \delta, t_c, t_s)$ combine those of $WTA(n, \delta',
t_c, t_s)$ and $Filter(n, \delta'')$.
That is, for any input sequence $\beta_{in}$ in which the $x$ inputs
are stable,
$Attention(n, \delta, t_c, t_s)$  specifies that, with probability at
least $(1 - \delta')$, the $y$ outputs converge to a single firing
output corresponding to some firing $x$ input within time $t_c$, and this
configuration persists for time $t_s$.
$Attention(n, \delta, t_c, t_s)$  also specifies that, with probability at least
$(1 - \delta'')^{t_s}$, the $z$ outputs always exhibit correct And behavior
with respect to the previous time's $y$ and $w$ firing behavior.
Together, these two properties imply that, assuming stable $x$ inputs,
with probability at least
$(1 - \delta) = (1 - \delta') (1 - \delta'')^{t_s}$, the $Attention$
network produces stable behavior of the part of the $y$ outputs, and
moreover, during the stable interval, the network produces $z$ outputs that
correctly mirror the $w$ inputs corresponding to the chosen $y$
output.

\hide{
Furthermore, assuming such stable behavior on the part of the $y_i$
outputs,  $Attention(n, \delta, t_c, t_s)$  specifies that, with probability at least
$(1 - \delta'')^{t_s}$, at all times during the stable interval, the $z_i$
outputs exhibit correct And behavior with respect to the previous time's $y_i$
and $w_i$ firing behavior.
%
Together, these properties imply that, with probability at least 
$(1 - \delta) = (1 - \delta') (1 - \delta'')^{t_s}$, the $Attention$
network correctly mirrors the inputs corresponding to the chosen $y_i$
output throughout the stable interval.
}

Theorem~\ref{th: composition-problems} implies that any compatible
solutions to $WTA(n, \delta', t_c, t_s)$ and $Filter(n, \delta'')$
can be composed to yield a solution to the composed problem
$Attention(n, \delta, t_c, t_s)$.
In particular, the solutions to these problems that we presented in
Sections~\ref{example: WTA-1} and~\ref{example: WTA-2} can be composed
in this way.

We can also define a version of the Attention problem in which we
hide the $y$ outputs, formally,
$hide(Attention(n, \delta, t_c, t_s),\{y_1,\ldots,y_n\})$.
The guarantees specified by this problem are similar to those of the
$Attention(n, \delta, t_c, t_s)$  problem, except that the behavior of
the $y$ neurons is not mentioned explicitly.
Essentially, this problem says that, with probability at least 
$(1 - \delta) = (1 - \delta') (1 - \delta'')^{t_s}$, the network
correctly mirrors the inputs corresponding to some $y$ output,
throughout the stable interval.
The same composition of solutions as above, with hiding of the $y$ outputs,
solves this version of the problem.


\section{Conclusions}
\label{sec: conclusions}

In this paper, we have presented a formal, mathematical foundation for
modeling and reasoning about the behavior of synchronous, stochastic
Spiking Neural Networks.
This foundation is based on a simple version of the SNN model in which
a neuron's only state is a Boolean value indicating whether the neuron
is currently firing.
We have provided definitions for networks and their externally-visible
behavior.
We have defined composition and hiding operators for building new SNNs
from others, and have proved fundamental theorems saying that these
operators preserve externally-visible behavior.
We have also defined a formal notion of a problem to be solved by an
SNN, and have given basic results showing how the composition and
hiding operators affect the problems that are solved by networks.

Future work will include using this formal foundation as a basis for
describing and verifying properties of particular SNNs.
We have already carried out rather formal proofs for some of our brain
network algorithms (see, e.g.,~\cite{LynchMP-arxiv19}).
However, these have been done in terms of models that were
specially-tailored to the problem at hand, and not in terms of a general
modeling framework; we believe that working in terms of a general
framework will contribute toward building a coherent general theory
for SNN algorithms.
A good starting point for such applications might be a study of
brain-like mechanisms for focusing attention, based on simpler
mechanisms such as our Winner-Take-All and Filter networks.


In the basic SNN model used in this paper, each neuron has a state
that is just a Boolean indicating whether or not it is currently firing.
We plan to extend the definitions and results to
allow a neuron to have more elaborate state.
For example, as in~\cite{SuCL-jour19}, a neuron's state might
include history of its recent incoming potential or recent firing behavior.
Also, as in~\cite{LynchMallmannTrenn21}, we may want to allow a neuron's state
to include some Boolean flags that may turn the neuron on or off for performing
certain activities, such as learning; in the neuroscience literature,
such mechanisms are known as ``eligibility traces'' \cite{GerstnerLLCB}.
It remains to carefully extend the definitions and results in this
paper to these more elaborate cases; this paper should provide
a useful blueprint for these extensions.
With such model extensions in hand, it will be interesting to
revisit work by Valiant, Navlakha, Papadimitriou, and their
collaborators, such
as~\cite{Valiant00,DasguptaSN17,PapadimitriouVempala-itcs19}), trying
to recast it in terms of our general concurrency theory framework.

\bibliography{compositional-model}{}
\bibliographystyle{splncs04}

\end{document}